\documentclass[11pt]{amsart}

\usepackage{color}
\usepackage[english]{babel}
\usepackage[T1]{fontenc}
\usepackage{url}
\usepackage{graphicx,colortbl}
\usepackage[colorlinks=true,citecolor=magenta,linkcolor=blue]{hyperref}
\definecolor{lightgray} {gray}{0.7}

\newcommand{\R}{\mathbb R}
\newcommand{\N}{\mathbb N}

\def\be#1\ee{\begin{equation}#1\end{equation}}
\newcommand{\fer}[1]{(\ref{#1})}

\newcommand{\bq}{\begin{equation}}
\newcommand{\eq}{\end{equation}}
\def\bqa{\begin{eqnarray}}
\def\eqa{\end{eqnarray}}

\def\e{\epsilon}

\newcommand{\sign}{\mathrm{sign}}
\newcommand{\dd}{\mathrm{d}}
\newcommand{\bd}{\begin{displaymath}}
\newcommand{\ed}{\end{displaymath}}
\newcommand{\ba}{\begin{eqnarray}}
\newcommand{\ea}{\end{eqnarray}}

\def\N{\mathbb{N}}
\def\R{\mathbb{R}}

\def\CI {\mathcal{I}}

\def\e{\varepsilon}
\newtheorem{remark}{Remark}
\newtheorem{lemma}{Lemma}
\newtheorem{proposition}{Proposition}

\usepackage{verbatim}

\newenvironment{equations}{\equation\aligned}{\endaligned\endequation}
\begin{document}

%\markboth{A. Bondesan, G. Toscani, M. Zanella}{KINETIC COMPARTMENTAL MODELS DRIVEN BY OPINION DYNAMICS}

%%%%%%%%%%%%%%%%%%% Publisher's Area please ignore %%%%%%%%%%%%%%%%%%%%%%%
%
%\catchline{}{}{}{}{}
%
%%%%%%%%%%%%%%%%%%%%%%%%%%%%%%%%%%%%%%%%%%%%%%%%%%%%%%%%%%%%%%%%%%%%%%%%%%

\title[Kinetic compartmental models for vaccine hesitancy]{KINETIC COMPARTMENTAL MODELS DRIVEN BY\\ OPINION DYNAMICS: VACCINE HESITANCY \\ AND SOCIAL INFLUENCE
%\footnote{For the title, try not to 
%use more than 3 lines. Typeset the title in 10 pt 
%Times roman, uppercase and boldface.} 
}
%\titlerunning{}

\author[A. Bondesan]{Andrea Bondesan}
\address{Andrea Bondesan \hfill\break
	Department of Mathematics, University of Pavia, 
	via Ferrata 5, 27100 Pavia, Italy}
\email{andrea.bondesan@gmail.com}

\author[G. Toscani]{Giuseppe Toscani}
\address{Giuseppe Toscani \hfill\break
	Department of Mathematics, University of Pavia and IMATI CNR, 
	via Ferrata 5, 27100 Pavia, Italy}
\email{giuseppe.toscani@unipv.it}

\author[M. Zanella]{Mattia Zanella}
\address{Mattia Zanella \hfill\break
	Department of Mathematics, University of Pavia, 
	via Ferrata 5, 27100 Pavia, Italy}
\email{mattia.zanella@unipv.it}

\maketitle

\vspace*{-0.8cm}
\begin{abstract}
We propose a kinetic model for understanding the link between opinion formation phenomena and epidemic dynamics. The recent pandemic has brought to light that vaccine hesitancy can present different phases and temporal and spatial variations, presumably due to the different social features of individuals. The emergence of patterns in societal reactions permits to design and predict the trends of a pandemic. This suggests that the problem of vaccine hesitancy can be described in mathematical terms, by suitably coupling a kinetic compartmental model for the spreading of an infectious disease  with the evolution of the personal opinion of individuals, in the presence of leaders. The resulting model makes it possible to predict the collective compliance with vaccination campaigns as the pandemic evolves and to highlight the best strategy to set up for maximizing the vaccination coverage. We conduct numerical investigations which confirm the ability of the model to describe different phenomena related to the spread of an epidemic.
\end{abstract}

\vspace*{0.5cm}

\vspace*{0.5cm}

\tableofcontents

%%%%%%%%%%%%%%%%%%%%%%%%%%%%%%%%%%%%%%%%%%%%%%%%%%%%%%%%%%%%
%%%%%%%%%%%%%%%%%%%%%%%%%%%%%%%%%%%%%%%%%%%%%%%%%%%%%%%%%%%%
%%%%%%%%%%%%%%%%%%%%    SECTION 1: INTRODUCTION    %%%%%%%%%%%%%%%%%%%%%%
%%%%%%%%%%%%%%%%%%%%%%%%%%%%%%%%%%%%%%%%%%%%%%%%%%%%%%%%%%%%
%%%%%%%%%%%%%%%%%%%%%%%%%%%%%%%%%%%%%%%%%%%%%%%%%%%%%%%%%%%%
\section{Introduction}

During the COVID-19 epidemic it was observed that, as cases of infection escalated, the collective adherence to so-called non pharmacological interventions (NPIs) was crucial to ensure public health in the absence of effective treatment \cite{APZ21,BC,BBDP,Vig,Z}. Indeed, the effectiveness of the measures proposed by policy-makers depends heavily on the compliance with them, which, in turn, depends on individuals' personal opinions about the necessity of social restrictions \cite{BRBU,DC,Tc,Tsao}. Similarly, as documented by Refs. \cite{Leu,Franceschi_etal,VBA}, vaccine hesitancy is deeply related to epidemic dynamics. 

Improving individual response to health protection is indeed a central issue in devising effective measures that lead to virtuous changes in daily social habits \cite{DFD}. To this extent,  it is important to study the close relationship between social features and the spreading of a pandemic \cite{BC,Bell,BB2,BRBU,DPTZ,DTZ,Tsao,Tu}, with the aim of coupling classical epidemiological models with opinion formation dynamics to understand the mutual influence of these phenomena in large many-agent systems that mimick real societies.\cite{GFPS,KGFPS,PC1,PC2,ZZ} 

The study of kinetic models for large interacting systems has gained increasing interest in recent years \cite{BB1,BCC,CFTV,CMPS,CPT,CPS,FHT,HT,MT}.  Among these, a prominent position has been taken by the phenomena of opinion formation, often described through the methods of statistical mechanics \cite{CFL,HK,SWS,W}.  A solid theoretical framework suitable for investigating emerging properties of opinion formation phenomena by means of mathematical models has been provided by classical kinetic theory \cite{CT,DMPW,DW,T,TTZ}.

 According to the kinetic description \cite{PT13}, the structure of the elementary variations of opinion  leads at the macroscopic scale to an  equilibrium distribution which describes the formation of a relative consensus about certain opinions \cite{PTTZ,T,TTZ}.  Analogously to what happens in kinetic theory of rarefied gases, where elastic collisions generate equilibrium distributions in the form of Maxwellian (Gaussian) densities, in opinion dynamics the equilibrium may assume the form of a Beta distribution \cite{FPTT,T}. In some cases,  deviation from global consensus  appears in the form of opinion polarization, describing a marked divergence away from central positions toward the extremes \cite{LRT}. This latter feature of the agents' opinion distribution is frequently observed in problems of choice formation \cite{ANT}.

Among other social aspects of opinion formation which deserve to be investigated in reason of their importance for individual health, vaccine hesitancy is one of the most interesting \cite{BDMOG,BRBU,DMLM,Tsao,Trentini}.  
This phenomenon has been shown to be an important feature in the evolution of the COVID-19 pandemic (cf.  Ref. \cite{Kum} and the references therein). In agreement with the analysis of Ref. \cite{Kum}, vaccine hesitancy can be classified as a phenomenon of clear social significance. Indeed, vaccine hesitancy has exhibited different phases, subject to both temporal and spatial variation, mainly due  to varied socio-behavioral characteristics of individuals and their response toward the COVID-19 pandemic and its vaccination strategies.  According to this picture, individuals in the society can be split in different groups, characterized by  their degree of susceptibility to vaccination:  Vaccine Eagerness, Vaccine Ignorance, Vaccine Resistance, Vaccine Confidence, Vaccine Complacency and Vaccine Apathy. The detailed description of these phases can be achieved by a careful reading of Ref. \cite{Kum}. Although the division of the population into different groups characterized by the type of vaccine hesitancy is very useful in order to best acquire the tools to limit its size, the boundaries between the groups are very blurred, as individuals belonging by convention to different groups may have very similar views regarding their desirability of vaccination. This convention is also present in other social phenomena, such as the distribution of wealth in Western societies, where the individuals are conventionally divided into classes (typically the poor class, the middle class and the rich class) and wealth varies continuously across classes.  

For this reason, it can be assumed that vaccine hesitancy is triggered by an opinion formation dynamics of each individual. We will assume that agents' opinion can be measured by a real value $w \in [-1,1]$, where the extremes $-1$ and $+1$ are associated respectively with complete unwillingness and complete willingness to undergo vaccination. Clearly, the value $w=0$ will correspond to vaccine apathy, or indifference to undergo vaccination. Concisely, in the rest of the paper we will define the state $w \in [-1,1]$ of an individual toward vaccination simply as \emph{opinion}. Next, following the well-known strategy of kinetic theory, we will classify the status of a population regarding the vaccination in terms of the statistical distribution of their opinions, say $f=f(w,t)$, where $w\in [-1,1]$ and $t \ge 0$.  If one agrees with this representation, various techniques of kinetic theory are available to study the evolution in time of the statistical distribution $f(w,t)$ \cite{PT13,T}.

According to the social description of opinion formation \cite{PT13,T}, changes in the  distribution $f$  over time are consequent to elementary interactions that mainly take into account the effects of compromise and self-thinking, which quantify the variation of the individual's opinion with respect to the group's opinion (the former) and the variation of opinion due to personal beliefs (the latter). In presence of a huge number of individuals undergoing elementary interactions, it can be shown that the statistical distribution $f=f(w,t)$ of opinions quickly relaxes in time toward a stationary profile, that in various relevant cases takes the form of a Beta distribution. 

Unlike the classical case where it is assumed that the mean opinion of the society does not vary with time, we want to study a situation in which the mean opinion about vaccination is subject to  temporal  variations, to account for the evolving socio-behavioral characteristics of individuals and their response toward the pandemic spreading and the vaccination strategies.

Following the recent contributions of Refs. \cite{DPaTZ,DPTZ,DTZ,Z}, this goal will be achieved by studying the evolution of the distribution function  of the vaccine hesitancy of the population  with a classical compartmental model for the spreading of a disease \cite{DH}. This coupling will act in both directions. On the one side, the presence of the disease will lead in general to a decrease of the vaccine hesitancy  of individuals (moving $w$ to the right), because of the perceived increase in risks associated with exposure to the infection \cite{DC,VBA}, and because of possible pro-vaccination campaigns, characterized in the present work by the so-called opinion leaders \cite{DMPW,DW}. On the other side,  phases of stagnation or decline in the epidemic spreading naturally lead to a growth in the vaccine hesitancy of individuals (moving $w$ to the left), modeled here by introducing a dependence of the group's mean vaccination rate on the number of infected individuals.

The rest of the paper is organized as follows. In Section \ref{sec:social} we introduce a kinetic compartment model coupling epidemic and opinion formation dynamics in the presence of vaccination. The effects of leader--follower interactions are discussed in Section \ref{sec:leader}. The main properties of the resulting model are then discussed in Section \ref{sec:analysis}. Finally, in Section \ref{sec:num} we present several numerical tests based on the proposed modelling setting, to shade light on the links among the behavioral aspects that may influence a vaccination campaign.

%%%%%%%%%%%%%%%%%%%%%%%%%%%%%%%%%%%%%%%%%%%%%%%%%%%%%%%%%%%%
%%%%%%%%%%%%%%%%%%%%%%%%%%%%%%%%%%%%%%%%%%%%%%%%%%%%%%%%%%%%
%%%%%%%%%%%%%%%%%%    SECTION 2: BEHAVIORAL APPROACH    %%%%%%%%%%%%%%%%%%%
%%%%%%%%%%%%%%%%%%%%%%%%%%%%%%%%%%%%%%%%%%%%%%%%%%%%%%%%%%%%
%%%%%%%%%%%%%%%%%%%%%%%%%%%%%%%%%%%%%%%%%%%%%%%%%%%%%%%%%%%%
\section{Kinetic models in compartmental epidemiology}\label{sec:social}
 
Classical epidemiological models rely on a compartmentalization of the population, in the following denoted by the ordered set $\mathcal C=\{J_1,\dots,J_n\}$, $n \ge 1$, in which the agents composing the system are split into epidemiologically relevant states, meaning that each state contributes to determine the evolution of the epidemic, for example the susceptible individuals who can contract the disease, or the infected and infectious agents. 
 
Given a statistical understanding on the agents' behaviour, a possible strategy to determine the transition rates between compartments relies on kinetic theory. Indeed, thanks to the methods of kinetic models for social phenomena, the emergence of collective structures and patterns determines at the observable level the evolution of the epidemic (cf. Ref. \cite{ABBDPTZ} and the references therein).

In the kinetic approach, agents in the compartment $J_i \in \mathcal C$, $1 \leq i \leq n$, are characterized by a social variable denoted by $w\in  \Omega \subset \R$ (number of daily  contacts, opinion, $\dots$), whose statistical description is obtained through the distribution function $f_{J_i}(w, t)$, $t \ge 0$. This means that
 $f_{J_i} = f_{J_i}(w,t)$ is such that $f_{J_i}(w,t)\dd w$ represents the fraction of agents with social variable in $[w,w+\dd w]$ at time $t\ge 0$ belonging to the $i$-th compartment. Furthermore, it is usually assumed that 
\[
\sum_{i=1}^n f_{J_i}(w,t) = f(w,t), \qquad  \int_{\Omega} f(w,t) \dd w = 1.
\]
The description of the time evolution of the disease in terms of the social variable $w$ is obviously richer than the classical description. Indeed, the knowledge of the distribution functions  $f_{J_i}(w, t)$, $J_i \in \mathcal C$, allows to compute, in addition to the mass fractions of the population in each compartment,  the  relevant moments of order $r \in \N$, that are
\begin{equation}
\label{eq:rho-m}
\rho_{J_i}(t) = \int_\Omega f_{J_i}(w,t)\dd w, \qquad  m_{r,J_i}(t) = \frac{1}{\rho_{J_i}(t)} \int_\Omega w^r f_{J_i}(w,t) \dd w. 
\end{equation}

As proposed in Ref. \cite{DPaTZ},  the time evolution of the functions $f_{J_i}(w,t)$ is obtained by coupling the compartmental epidemiological description with the kinetic equilibration of the social variable. Let $\textbf{f}(w,t)$ denote the $n$-dimensional vector with components the functions $f_{J_i}(w,t)$. This merging can then be concisely expressed  by the system
\be\label{kin1}
\frac{\partial\textbf{f}(w,t)}{\partial t} = \textbf{P}(w, \textbf{f}(w,t)) + \frac 1\tau \textbf{Q}(\textbf{f}(w,t)).
\ee
where  $\textbf Q(\cdot,\cdot)$ is a $n$-dimensional vector whose components $Q_{J_i}(\textbf{f})$, $J_i \in \mathcal C$, are suitable relaxation operators describing the formation of the equilibrium distribution of the social variable in the underlying compartments. Moreover, $\tau > 0$ measures the time-scale needed to reach these local equilibrium distributions in the relaxation process.

We underline that the model we propose in this work is truly multiscale, since it interfaces the dynamics of the epidemic with the one of the individuals which, in principle, are different. Indeed, the time-scale at which individuals adapt their social behavior to the presence of the pandemic is much faster than the typical time taken by the pandemic itself to evolve.

%%%%%%%%%%%%%%%%%    SUBSECTION: CONSENSUS FORMATION    %%%%%%%%%%%%%%%%%%
\subsection{Consensus formation and epidemic dynamics}\label{sec:consensus}
 
In this section we will focus on the case where the social variable is represented by the personal opinion of individuals.\cite{ZZ} Without loss of generality,  we will concentrate on a SEIRV compartmentalization of the population, in which the system of agents is subdivided in the following states: susceptible (S) agents, that can contract the disease; infectious (I)  agents, responsible for the spread of the disease; exposed (E) agents, who have been infected but are still not contagious; vaccinated (V) agents, who received the vaccine and are now immune; and, finally, removed (R) agents, who cannot spread the disease. Each agent is endowed with an opinion variable $w \in \mathcal{I}$ which varies continuously in the interval $\mathcal{I} = [-1,1]$, where $-1$ and $+1$ denote the two opposite beliefs regarding the protective behavior.  In particular, the value $w = -1$ is associated with agents that do not believe in the necessity of protections (like wearing masks or reducing daily contacts), whereas $w = 1$ is associated with agents that are in complete agreement with a protective behavior. It will be also assumed that agents characterized by a high protective behavior are less likely to get the infection. 

Following the notations of Section \ref{sec:social},  we denote by $f_J(w,t)$ the distribution of opinions at time $t \ge 0$ of the agents in the compartment $J \in \mathcal C = \{S,E,I,R,V\}$. \\
 According to Ref. \cite{ZZ}, the explicit form of the kinetic system \fer{kin1} for the coupled evolution of opinions and disease is given by
 \begin{equation}
\label{eq:kinetic_opinion}
\begin{split}
\frac{\partial f_S(w,t)}{\partial t} &= -K_T(f_S,f_I)(w,t) - K_V(f_S) + \dfrac{1}{\tau} Q_{S}(f_S)(w,t), \\
\frac{\partial f_E(w,t)}{\partial t} &= K_T(f_S,f_I)(w,t) - \sigma_E f_E(w,t) + \dfrac{1}{\tau}Q_{E}(f_E)(w,t), \\
\frac{\partial f_I(w,t)}{\partial t} &= \sigma_E f_E(w,t) - \gamma_I f_I(w,t) + \dfrac{1}{\tau} Q_{I}(f_I)(w,t), \\
\frac{\partial f_R(w,t)}{\partial t} &= \gamma_I f_I(w,t) + \dfrac{1}{\tau}Q_{R}(f_R)(w,t), \\
\frac{\partial f_V(w,t)}{\partial t} &=  K_V(f_S) + \dfrac{1}{\tau} Q_{V}(f_V)(w,t),
\end{split}
\end{equation}
In system \fer{eq:kinetic_opinion}, $\tau > 0$ is the relaxation time and $Q_{J}(\cdot)$ is the kinetic operator that quantifies the variation of the opinions of agents in the compartment $J \in \mathcal C$. The precise form of these operators will be specified in the next Section \ref{sec:opinion}. The parameter $\sigma_E>0$ is such that $1/\sigma_E$ measures the mean latent period of the infection, whereas $\gamma_I > 0$ is such that $1/\gamma_I>0$ is the mean infectious period \cite{DH}. 

We have denoted with $K_T(\cdot,\cdot)$ the local incidence rate governing the transmission dynamics, that has the form
\begin{equation}
\label{eq:K_def}
K_T(f_S,f_I)(w,t) =  f_S(w,t) \int_{\mathcal{I}} \kappa_T(w,w_*) f_I(w_*,t)dw_*.
\end{equation}
This operator drives the transmission of the infection in terms of the opinion variable $w$. In \fer{eq:K_def}, $\kappa_T(w,w_*)$ is a nonnegative decreasing function measuring the impact of the protective behavior on the interactions between susceptible and infectious agents. A leading example for the function $\kappa_T(w,w_*)$ can be obtained by assuming 
\begin{equation}
\label{eq:kappa_def}
\kappa_T(w,w_*) = \dfrac{\beta_T}{4^{\alpha_T}} (1-w)^{\alpha_T} (1-w_*)^{\alpha_T},
\end{equation}
where $\beta_T > 0$ is the baseline transmission rate characterizing the epidemics and $\alpha_T > 0$ is a coefficient linked to the efficacy of the protective measures. Thus, the local rate of incidence, and consequently the evolution of the disease, is fully dependent on the joint opinion  of agents with respect to the protective behavior. In Figure \ref{fig:kap}, the function $\kappa_T(\cdot,\cdot)$ in the particular form \eqref{eq:kappa_def} is represented for several values of $\alpha_T > 0$. 

Furthermore, in  \eqref{eq:kinetic_opinion} the function $K_V(\cdot)$ denotes the local vaccination rate, which characterizes the vaccination dynamics in terms of the opinion variable and has the form 
\begin{equation}\label{eq:KVdef}
K_V(f_S)(w,t) = f_S(w,t)  \kappa_V(w), 
\end{equation}
where $\kappa_V(\cdot)$ is a nonnegative increasing function that expresses the vaccination propensity of the susceptible agents with opinion $w \in [-1,1]$. A prototypical example is obtained with 
\begin{equation}
\label{eq:kappav_def}
\kappa_V(w) = \frac{\beta_V}{2^{\alpha_V}}(1+w)^{\alpha_V},
\end{equation}
which gives 
\begin{equation}
K_V(f_S)(w,t) = \dfrac{\beta_V}{2^{\alpha_V}} f_S(w,t)(1+w)^{\alpha_V},
\end{equation}
where $\beta_V>0$ is the baseline vaccination rate and $\alpha_V>0$ is a coefficient measuring the impact of the epidemic on the vaccination campaign. 

\begin{figure}
\centering
\includegraphics[scale = 0.3]{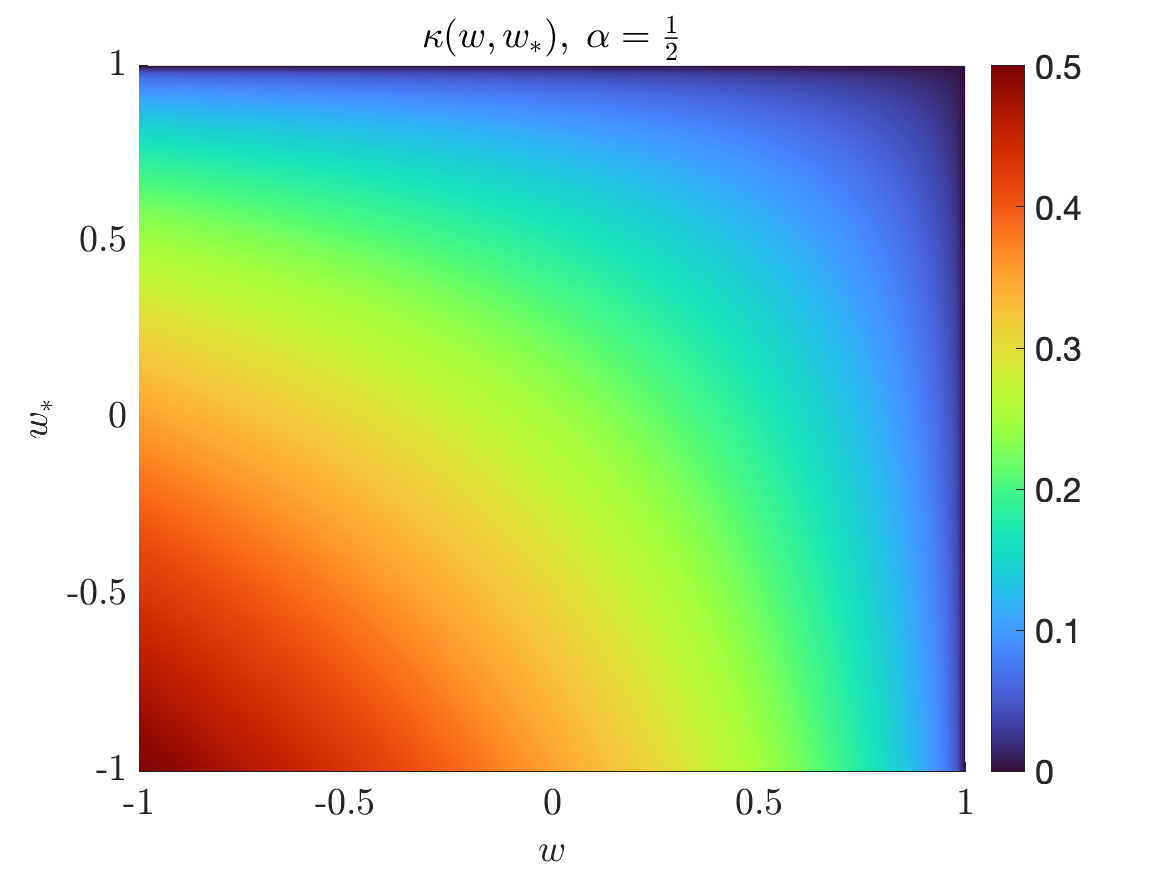}
\includegraphics[scale = 0.3]{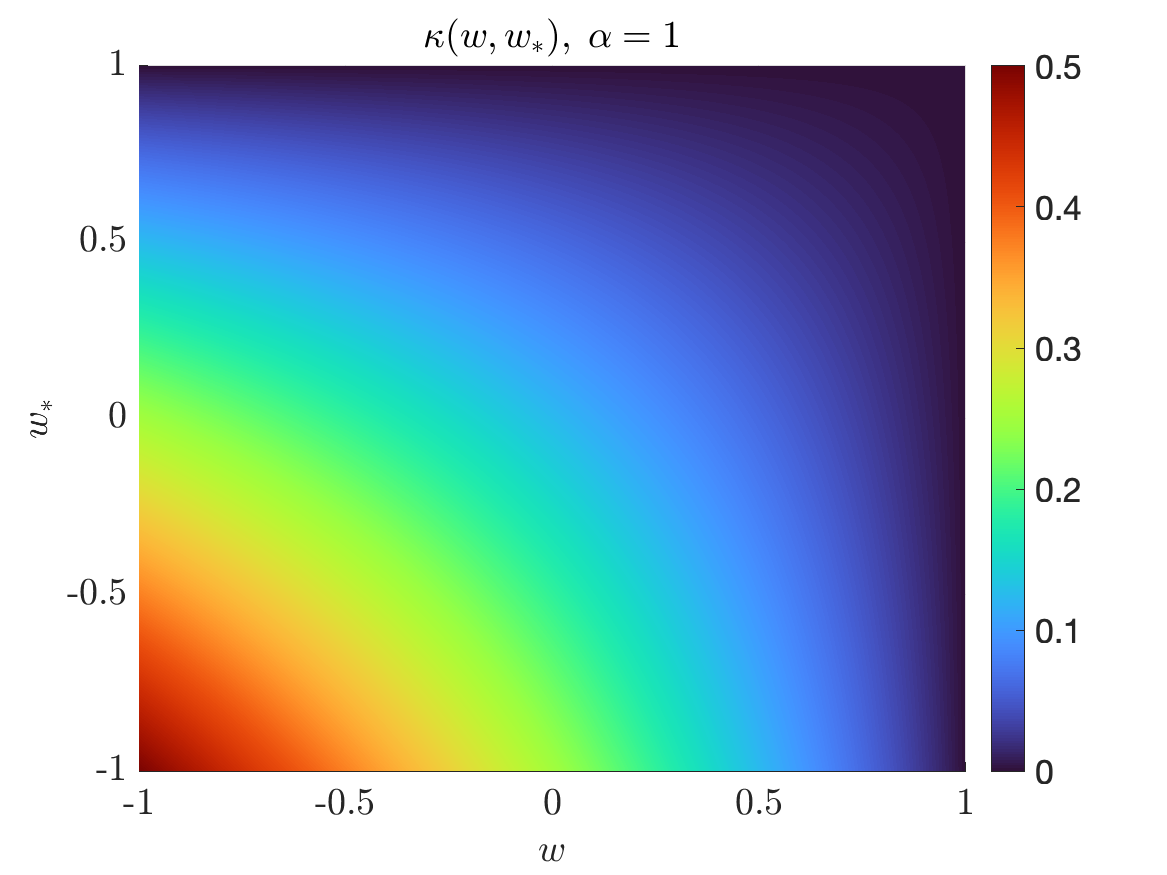}
\caption{We sketch the function $\kappa_T(w,w_*)$ in \eqref{eq:kappa_def} for $\alpha_T=\frac{1}{2}$ (left) and $\alpha_T = 1$ (right). In both cases, we fixed the coefficient $\beta_T = \frac{1}{2}$. }
\label{fig:kap}
\end{figure}

%%%%%%%%%%%%%%%%%%%%    SUBSECTION: OPINION FORMATION   %%%%%%%%%%%%%%%%%%%
\subsection{Kinetic modelling in opinion dynamics}\label{sec:opinion}

In this section we introduce the kinetic operators $Q_J(\cdot)$, $J \in \mathcal C$, which determine the relaxation of the kinetic density $f_J$ toward its equilibrium. To this end, we seek to identify a model-based scenario for the evolution of opinion-type dynamics in the presence of an epidemic. In particular, we intend to describe the evolution of the vaccine hesitancy of a population by resorting to classical methods of statistical mechanics, and in particular by emphasizing its analogies with the  process of opinion formation, well studied  within the toolbox of classical kinetic theory \cite{PT13}. In this way, one will be able to present an almost uniform description of vaccine hesitancy in terms of few simple rules. The kinetic modeling of opinion formation in terms of Fokker--Planck-type equations was introduced  in 2006 in Ref. \cite{T}, and was subsequently generalized in many ways (see Refs. \cite{NPT,PT13} for recent surveys).

Kinetic-type models for opinion formation are classically built under the fundamental assumption of indistinguishable agents \cite{PT13}, which corresponds to assume that each agent of the system is characterized by its opinion $w \in \CI$. Then, the evolution of the density of agents in each compartment $f_J = f_J(w, t)$, $w \in \CI$, $t \ge 0$, is determined by repeated interactions. In the following we will assume that each agent modifies its opinion through a simple interaction with an opinion distribution acting as a background. We denote by $g(z)$ the static background opinion distribution 
\be\label{bg}
\int_{\mathcal I}g(z)\dd z = 1, \qquad \int_{\mathcal I} zg(z)\dd z = \mu \in \mathcal I. 
\ee
 Following Ref. \cite{T}, the microscopic change of opinion is modeled as
\begin{equation}\label{colli-op}
w^\prime = w + \lambda_J (z-w) + D(w) \eta_J, 
\end{equation}
where $\lambda_J>0$, $z$ is a realization of the background opinion distribution $g(z)$, while the $\eta_J$, $J \in \mathcal{C}$, are centered random variables such that $\left \langle \eta_J^2 \right\rangle = \sigma_J^2$. Furthermore, in \fer{colli-op} we have introduced the constant $\mu \in \CI$, which denotes a background average opinion. The constant $\lambda_J$ and the variance $\sigma_J^2$ of each random variable $\eta_J$ measure respectively the propensity to move toward the average opinion in society (the compromise) and the degree of spreading of opinion because of diffusion, which describes possible changes of opinion due for example to personal access to information (self-thinking). Finally, the function $D(\cdot)$ takes into account the local relevance of self-thinking.  Following Ref. \cite{T}, we may assume
\be\label{di}
D(w) =\sqrt{1-w^2},
\ee
 By classical methods of kinetic theory \cite{PT13}, and resorting to the derivation of the classical linear Boltzmann equation of elastic rarefied gases \cite{Cer}, one can easily show that  the time variation of the opinion density depends on a sequence of elementary  variations of type \fer{colli-op}.  In our case, the  density $f_J(w,t)$ of each agents' compartment 
obeys, for all smooth functions $\varphi(w)$ (the observable quantities), to the linear integro-differential equation 
 \begin{equation} \label{line-op}
 \frac{\dd}{\dd t}\int_{\mathcal{I}}\varphi(w) f_J(w,t)\,\dd w  = \frac 1{\tau}
 \int_{\CI}    \big\langle \varphi(w^\prime)-\varphi(w) \big\rangle f_J(w,t)g(z)\,\dd w \dd z,
 \end{equation}
 where $\tau > 0$ is a suitable relaxation parameter.  In \fer{line-op} the mathematical expectation $\langle \cdot \rangle$  takes into account the presence of the random variable $\eta_J$ in \fer{colli-op}. Then,  the \emph{quasi-invariant opinion limit} is obtained by scaling all the parameters in the kinetic equation \fer{line-op} through a parameter $\e \ll 1$ as
  \[
  \sigma_J^2 \to \e \sigma_J^2, \quad \lambda_J \to \e \lambda_J, \quad \tau \to \e \tau,
  \]
which implies that the interaction \fer{colli-op} produces only an extremely small variation of the post-interaction opinion, while the frequency of interactions is increased accordingly. Note that in this scaling, the evolution of the mean  opinion value does not depend on $\e$. Also, the same property holds for the constant $\nu_J = \sigma_J^2/\lambda_J$. 
 As exhaustively explained in Ref. \cite{FPTT17}, this asymptotic procedure is a well-consolidated technique which has been first developed for the classical Boltzmann equation \cite{Vi,vil2}, where it is known under the name of \emph{grazing collision limit}. Then, taking the limit $\e \to 0$ \cite{PT13,T}, the solution to the kinetic equation converges to the solution of a Fokker--Planck-type equation defined as follows
 \begin{equation}\label{eq:FPmu}
 \dfrac{\partial  f_J(w,t)}{\partial t} = \dfrac{1}{\tau}\left[ \dfrac{\nu_J}{2} \dfrac{\partial^2}{\partial w^2}((1-w^2) f_J(w,t)) + \dfrac{\partial}{\partial w}(w-\mu) f_J(w,t)\right], 
 \end{equation}
 where $\nu_J = \sigma_J^2/\lambda_J$. The Fokker--Planck equation \fer{eq:FPmu} is then complemented by no-flux boundary conditions, for any $t>0$,
 \[\begin{split}
& (w-\mu)f_J(w,t) + \dfrac{\nu_J}{2}\dfrac{\partial}{\partial w}((1-w^2)f_J(w,t)) \Big|_{w = \pm 1} = 0, \\
 & (1-w^2)f_J(w,t) \Big|_{w = \pm 1} = 0. 
 \end{split}\]
 It is important to notice that the parameter  $\sigma_J^2$, namely the variance of the random variable $\eta_J$ measuring the compartmental self-thinking, appears in equation \fer{eq:FPmu} as the coefficient of the diffusion operator, while the parameter $\lambda_J$ measuring the compromise is the coefficient of the drift term. Consequently, small values of the parameter $\nu_J$ correspond to \emph{compromise dominated} dynamics, while large values of $\nu_J$ denote  \emph{self-thinking dominated} dynamics. 
  
In particular, the equilibrium distribution is explicitly computable for any $J \in \mathcal{C}$ and is given by the Beta densities
\begin{equation} \label{equ}
f^{\infty}_J(w) = C_{\mu,\nu_J} (1 - w)^{-1 + \frac{1 - \mu}{\nu_J}} (1 + w)^{-1 + \frac{1 + \mu}{\nu_J}},
 \end{equation}
where $C_{\mu,\nu_J}>0$ is a suitable normalization constant.

It is important to underline that a marked advantage of the grazing procedure is the ability to explicitly compute the macroscopic equilibrium  density \fer{equ} and its principal moments in dependence of the characteristic constants $\mu$ and $\nu_J$.

%%%%%%%%%%%%%%%%%%    SUBSECTION: VACCINE HESITANCY    %%%%%%%%%%%%%%%%%%%
 \subsection{Vaccine hesitancy and opinion formation}\label{sec:hesitancy}
  
 In this section, we couple the opinion formation dynamics defined in Section \ref{sec:opinion} with the evolution of an epidemic. To this extent, we generalize the microscopic interactions with the environment by suitably modifying the average opinion $\mu$ of the background distribution \fer{bg}, which is now dependent on the number of infected individuals at time $t\ge0$, given by  $\rho_I(t) = \int_{\mathcal I}f_I(w,t)\dd w$. In particular, denoting by $g(z,t)$ the new evolving background opinion distribution, we assume that, for any $t \geq 0$,
\[
\int_{\mathcal I}g(z,t)\dd z = 1, \qquad \int_{\mathcal I} z g(z,t) \dd z = \mu(\rho_I(t)) \in \mathcal I,
\]
This can be easily obtained by setting
\[
g(z,t)  = \frac 1{A(\rho_I(t))} g\left(\frac z{A(\rho_I(t))} \right), 
\] 
where  $\mu A(s) = \mu(s)$ is given by
 \be\label{risk}
 \mu(s) = 1 - 2(s - 1)^{2 n},
 \ee
with $s \in [0,1]$ and $n \in \mathbb N$ being a positive constant. Under the introduced assumption, we may observe that $\mu(s)$ is an increasing function and $\mu\equiv -1$ if $s = 0$, i.e. in the absence of infected individuals, whereas $\mu \equiv 1$ if $s = 1$, i.e. in the case where the infected individuals represent the whole population. This can be interpreted as a reaction of the individuals to the current evolution of the epidemic, so that the presence of many infected leads the population to protective behaviors (since in this case $\mu$ tends to $1$) in an effort to reduce the cases, while a reduction in the epidemic spreading pushes the population toward more risky decisions as the compartmental mean now is driven by $\mu$ to $-1$. In particular, the constant $n$ is related to the intensity of the risk perception. Indeed, as observed in Ref. \cite{VBA},  a key factor that moves people to vaccinate  is the perception of risks associated with the infection.  Risk perception is central to most health behavior theories and in this line, a recent analysis   based on  15,000 individuals, presented in  Ref. \cite{PB}, revealed a significant association between vaccine uptake and the perceived likelihood of infection. 

According to the discussion of Section \ref{sec:opinion}, the distribution function of the vaccine hesitancy of individuals, say $f_J(w,t)$, with $J \in \mathcal C$ denoting the various compartments of the population, is now driven by the Fokker--Planck equation
\be\label{op-FPS}
 \frac{\partial f_J(w,t)}{\partial t} = \frac{\nu_J}{2} \frac{\partial^2 }{\partial w^2}\left((1-w^2)
 f_J(w,t)\right) + \frac{\partial }{\partial w}\left((w - \mu(\rho_I))f_J(w,t)\right),
 \ee
with $\nu_J = \sigma_J^2/\lambda_J$, coupled with the no-flux boundary conditions, for any $t > 0$,
\begin{equation} \label{op-FPS-BC}
\begin{aligned}
&(w - \mu(\rho_I)) f_J(w,t) + \frac{\nu_J}{2} \frac{\partial}{\partial w}\left((1-w^2) f_J(w,t)\right) \Big|_{w = \pm 1} = 0, & \\
&(1-w^2) f_J(w,t) \Big|_{w = \pm 1} = 0.
\end{aligned}
\end{equation}

\begin{remark}
We highlight that equation \eqref{op-FPS} is mass-preserving, so that the quantities $\rho_J(t)$ remains constant over time, for any $J \in \mathcal{C}$. Since this is true in particular for the number of infected individuals $\rho_I(t)$, in this part $\mu(\rho_I)$ simply denotes a constant value ranging in the interval $\mathcal{I}$. However, we underline that the temporal evolution of $\rho_I(t)$ will be taken into account in our following analysis, where we introduce the full kinetic model \eqref{eq:kinetic_hesitancy}. There, mass conservation will be lost in the interaction of the Fokker--Planck opinion operators with the epidemic operators $K_T$, $K_V$ defined by \eqref{eq:K_def} and \eqref{eq:KVdef}, and the average background opinion $\mu(\cdot)$ will indeed vary over time depending on $\rho_I(t)$.
\end{remark}

The evolution of the first order moment $m_J(t) := m_{1,J}(t)$ is obtained from \fer{op-FPS}
\[
\begin{split}
&\dfrac{\dd m_J(t)}{\dd t} = \frac{1}{\rho_J} \dfrac{\dd}{\dd t} \int_{\mathcal I}w f_J(w,t)\dd w  \\[2mm]
&  = \frac{1}{\rho_J} \int_{\mathcal I} w \dfrac{\partial}{\partial w}(w-\mu(\rho_I))f_J(w,t)\dd w + \frac{1}{\rho_J}\dfrac{\nu_J}{2} \int_{\mathcal I} w  \dfrac{\partial^2}{\partial w^2}((1-w^2)f_J(w,t))\dd w  \\[4mm]
& = - (m_J(t) - \mu(\rho_I)),
\end{split}\]
therefore $m_J(t) \to \mu(\rho_I)$ for $t \to +\infty$ in \fer{op-FPS}. On the other hand, the evolution of the energy $E_J(t) := m_{2,J}(t)$ is given by 
\[
\begin{split}
&\dfrac{\dd E_J(t)}{\dd t} = \frac{1}{\rho_J}\dfrac{\dd}{\dd t} \int_{\mathcal I}w^2 f_J(w,t)\dd w  \\[2mm]
& = \frac{1}{\rho_J} \int_{\mathcal I} w^2 \dfrac{\partial}{\partial w}(w-\mu(\rho_I))f_J(w,t)\dd w + \frac{1}{\rho_J} \dfrac{\nu_J}{2} \int_{\mathcal I} w^2  \dfrac{\partial^2}{\partial w^2}((1-w^2)f_J(w,t))\dd w  \\[4mm]
& = - (2 + \nu_J) E_J(t) + 2 \mu(\rho_I) m_J(t) + \nu_J. 
\end{split}
\]
Hence, in the absence of fluctuations $\sigma_J^2\equiv 0$, we have that $\nu_J \equiv 0$ and $E_J(t) \to \mu^2(\rho_I)$ for $t \to +\infty$. As a consequence we can observe that, given a random variable $X \sim f_J(w,t)/\rho_J$, then
\[
\textrm{Var}[X] = \frac{1}{\rho_J} \int_{\mathcal I} (w-\mu(\rho_I))^2 f_J(w,t) \dd w \to 0, 
\]
for $t \to +\infty$.

%%%%%%%%%%%%%%%%%%%%%%%%%%%%%%%%%%%%%%%%%%%%%%%%%%%%%%%%%%%
%%%%%%%%%%%%%%%%%%%%%%%%%%%%%%%%%%%%%%%%%%%%%%%%%%%%%%%%%%%
%%%%%%%%%%%%%%%%%%    SECTION: MODELING THE LEADERS    %%%%%%%%%%%%%%%%%%%
%%%%%%%%%%%%%%%%%%%%%%%%%%%%%%%%%%%%%%%%%%%%%%%%%%%%%%%%%%%
%%%%%%%%%%%%%%%%%%%%%%%%%%%%%%%%%%%%%%%%%%%%%%%%%%%%%%%%%%%
\section{Impact of opinion leaders in kinetic compartmental modelling}
\label{sec:leader}

As far as vaccine hesitancy of a population is studied, other effects have to be considered. Among them, the relevance of a vaccination campaign. A vaccination campaign can easily be described in terms of the presence of opinion leaders. Opinion leaders can be defined as individuals that, in a binary exchange,  do not modify their opinion \cite{DMPW,DW}.
As discussed in Ref. \cite{DMPW}, in the last twenty years, new communication forms like email, web navigation, blogs and instant messaging have globally changed the way information is disseminated and opinions are formed within the society. Still, opinion leadership continues to play a critical role, independently of the underlying technology \cite{Rash}.

To characterize this type of interactions, we consider a leader--follower dynamics where agents belonging to the compartment $J \in \mathcal C$ interact with  opinion leaders. Assuming a pre-collisional opinion pair $(w, w_*)$ we obtain the post-interaction opinions $(w^{\prime \prime},w_*^{\prime \prime})$ defined  by the following anisotropic binary scheme
\begin{equation}\label{lea}
\begin{split}
w^{\prime \prime} &= w + \lambda_L G(w,w_*)(w_* - w) + D(w)\eta_L, \\
w^{\prime \prime}_* &= w_*,
\end{split}
\end{equation}
where $\lambda_L \in (0,1)$, $\langle \eta_L \rangle = \sigma_L^2$ and $G(\cdot,\cdot) \in [0,1]$ is an interaction function describing the influence of the leaders on the vaccine hesitancy of the population.  In the following, the local relevance of the diffusion modeled by $D(\cdot)$ and defined in \fer{lea} is the same as in \fer{colli-op}. We further assume that  for any $J \in \mathcal C$ we have $\lambda_J + \lambda_L= 1$. Thus, the values of $\lambda_J$ and $\lambda_L$  determine the percentage of change in the personal opinion due to interactions with the background and with the leaders, respectively.

We notice that since in \fer{lea} the contribution of the self-thinking may be different from the one present in \fer{colli-op}, the impact of self-thinking on individuals changes depending on whether they interact with the background or with the leaders. We further remark that it is reasonable to assume $w_* \approx 1$. Therefore, agents with an opinion $w \ge w_*$ interact more frequently with the leaders. A possible choice to model this behavior is given by the bounded confidence interaction function
\be\label{BC}
G(w,w_*) = \chi \left( |w - w_*| \le \Delta \right),
\ee
where $\chi(\cdot)$ denotes the indicator function while $\Delta \in [0,2]$. The ultimate goal of the leaders is thus to steer the opinion of agents toward  $w_*$ so that the vaccination rate defined by $K_V(\cdot,\cdot)$ is maximized.

Let $f_L(w)$ be the opinion density function characterizing a leader, and let $m_L = \int_{-1}^1 w f_L(w) \dd w$ denote its mean. In the presence of an epidemic spreading and a vaccination campaign, we can then model the evolution of the individuals, described by the vector-distribution function $(f_J(w,t))_{J \in \mathcal{C}}$ for the different compartments, interacting with a leader, characterized by the distribution $f_L(w)$, through Boltzmann-like equations which read in weak form, for any $J \in \mathcal{C}$, as
 \begin{equation*}
\begin{aligned}
\frac{\dd}{\dd t}\int_{\mathcal{I}}\varphi(w) f_J(w,t)\,\dd w = \frac 1{\tau} \int_{\mathcal{I}\times \mathcal I}  & \big\langle \varphi(w^\prime)-\varphi(w) \big\rangle f_J(w,t)g(z)\,\dd w \dd z\\[2mm]
& + \frac 1{\tau} \int_{\mathcal{I} \times \mathcal{I}}  \left\langle \varphi(w^{\prime \prime}) - \varphi(w) \right\rangle f_L(w_*) f_J(w,t) \,\dd w_* \dd w,
\end{aligned}
\end{equation*}
for any smooth function $\varphi(w)$, where the post-interaction opinions $w^\prime$ and $w^{\prime \prime}$ have been introduced respectively in \eqref{colli-op} and \eqref{lea}. Taking into account the quasi-invariant opinion limit described in Section \ref{sec:opinion} and assuming that $\lambda_J + \lambda_L = 1$, the vector-distribution function $(f_J(w,t))_{J \in \mathcal{C}}$ is then driven by the system of Fokker--Planck equations
\begin{equations}\label{op-FPH}
 \frac{\partial f_J(w,t)}{\partial t} & =  Q_J(f_J, f_L)(w,t) \\[4mm]
 & = \frac{\sigma_J^2 + \sigma_L^2}{2}\frac{\partial^2 }{\partial w^2}\left((1-w^2)f_J(w,t)\right)  \\[2mm]
 &  \hspace{2cm} + \frac{\partial }{\partial w}\left( \left( \lambda_J(w - \mu(\rho_I)) + \lambda_L \mathcal{L}[f_L](w,t) \right) f_J(w,t) \right),
 \end{equations}
where the linear operator $\mathcal{L}$ accounts for the interactions with the leader through the function $G$ and has the explicit form
\begin{equation} \label{op-leader}
\mathcal{L}[f_L](w) = \int_{\mathcal I}G(w,w_*) (w - w_*) f_L(w_*) \dd w_*,
\end{equation}
and the operators $Q_J$, $J \in \mathcal{C}$, satisfy, for any $t > 0$, the additional no-flux boundary conditions
\begin{equation} \label{op-FPH-BC}
\begin{aligned}
\left( \lambda_J(w - \mu(\rho_I)) + \lambda_L \mathcal{L}[f_L](w) \right) f_J(w,t) + \frac{\sigma_J^2 + \sigma_L^2}{2} \frac{\partial}{\partial w}\left((1-w^2) f_J(w,t)\right) \Big|_{w = \pm 1} = 0, & \\[4mm]
(1-w^2) f_J(w,t) \Big|_{w = \pm 1} = 0. &
\end{aligned}
\end{equation}
Therefore, in the general case, the drift term of the Fokker--Planck-type equations takes into account both contributions coming from the opinion changes of agents in the presence of the epidemic and the influence of leaders that promote protective behaviours, maximizing the vaccination rate. Note that, due to the nonlinear nature of the operator $\mathcal L[f_L]$  the explicit form of the equilibrium distribution is  difficult to obtain. However, assuming that the leaders' population is able to interact with the whole population, i.e. $\Delta = 2$ in \fer{BC}, we have  $G(\cdot,\cdot)\equiv 1$. In this case, owing to the constraint $\lambda_J + \lambda_L = 1$, the equilibrium has the explicit form
\be\label{pc}
f_J^{\infty}(w) = C_J (1 - w)^{-1 + \frac{1 - (\lambda_J \mu(\rho_I) + \lambda_L m_L)}{\sigma^2}} (1 + w)^{-1 + \frac{1+ (\lambda_J \mu(\rho_I) +\lambda_L m_L)}{\sigma^2}},
 \ee
where $\sigma^2 = \sigma_J^2 + \sigma_L^2$ and the constant $C_J = C_J(\mu(\rho_I), m_L, \lambda_J, \lambda_L, \sigma^2)>0$ is a normalization constant. The evolution of the first moment can be obtained from \fer{op-FPH} by integration
\[
\begin{split}
\dfrac{\dd m_J(t)}{\dd t} & = \frac{1}{\rho_J} \int_{\mathcal I} w \dfrac{\partial}{\partial w} \left( \lambda_J(w-\mu(\rho_I)) + \lambda_L(w-m_L)\right)f_J(w,t)\dd w \\
& \hspace{3.5cm} + \frac{1}{\rho_J}  \dfrac{\sigma^2}{2} \int_{\mathcal I} w \dfrac{\partial^2}{\partial w^2}((1-w^2)f_J(w,t))\dd w \\
& = - \frac{1}{\rho_J} \int_{\mathcal I}(w-(\lambda_J\mu(\rho_I) + \lambda_L m_L))f_J(w,t)\dd w,
\end{split}
\]
from which we get $m_J(t) \to \lambda_J \mu(\rho_I) + \lambda_L m_L$ for $t \to +\infty$. Furthermore, the evolution of the energy is given by 
\[
\begin{split}
\dfrac{\dd E_J(t)}{\dd t} & = \frac{1}{\rho_J} \int_{\mathcal I}w^2 \dfrac{\partial}{\partial w} \left( \lambda_J(w-\mu(\rho_I)) + \lambda_L(w-m_L)\right)f_J(w,t)\dd w \\
& \hspace{3.5cm} + \frac{1}{\rho_J} \dfrac{\sigma^2}{2} \int_{\mathcal I}w^2 \dfrac{\partial^2}{\partial w^2}((1-w^2)f_J(w,t))\dd w \\[2mm]
& = - (2 + \sigma^2) E_J(t) +  2 (\lambda_J\mu(\rho_I) + \lambda_L m_L) m_J(t) + \sigma^2,
\end{split}
\]
Hence, if $\sigma^2 \equiv0$ we have that a random variable $X\sim f_J(w,t)/\rho_J$ is such that $\textrm{Var}[X]\to 0$ for $t\to+ \infty$.

%%%%%%%%%%%%%%%%%%    SUBSECTION: FULL KINETIC MODEL    %%%%%%%%%%%%%%%%%%%
\subsection{The full opinion-based kinetic model}\label{sec:leaders}

We are now ready to couple the epidemic dynamics with the opinion formation model. The full kinetic epidemic system reads
\begin{equation}
\label{eq:kinetic_hesitancy}
\begin{split}
\frac{\partial f_S}{\partial t} & = -K_T(f_S,f_I) - K_V(f_S) + \dfrac{1}{\tau} Q_S(f_S,f_L), \\[2mm]
\frac{\partial f_E}{\partial t} & = K_T(f_S,f_I) - \sigma_E f_E + \dfrac{1}{\tau} Q_E(f_E,f_L), \\[2mm]
\frac{\partial f_I}{\partial t} & = \sigma_E f_E - \gamma_I f_I + \dfrac{1}{\tau} Q_I(f_I,f_L), \\[2mm]
\frac{\partial f_R}{\partial t} & = \gamma_I f_I + \dfrac{1}{\tau} Q_R(f_R,f_L), \\[2mm]
\frac{\partial f_V}{\partial t} & =  K_V(f_S) + \dfrac{1}{\tau} Q_V(f_V,f_L),
\end{split}
\end{equation}
where $\tau > 0$ is a relaxation parameter, and the Fokker--Planck operators describing the leader--follower dynamics $Q_J(\cdot,\cdot)$, $J\in \mathcal C = \{ S, E, I, R, V \}$, have been defined by \eqref{op-FPH}--\eqref{op-FPH-BC}. As in system \eqref{eq:kinetic_opinion},  the parameter $\sigma_E>0$ is such that $1/\sigma_E$ measures the mean latent period for the disease, whereas $\gamma_I>0$ is such that $1/\gamma_I>0$ measures the mean infectious period \cite{DH}. In \eqref{eq:kinetic_hesitancy}, the transmission of the infection is governed by the local incidence rate $K_T(f_S,f_I)$, defined in \fer{eq:K_def}, while the \emph{response} of the susceptible agents to the infection through vaccination is modeled by the operator $K_V(f_S)$, defined in \fer{eq:KVdef}.

%%%%%%%%%%%%%%%%%    SUBSECTION: MATHEMATICAL ANALYSIS    %%%%%%%%%%%%%%%%%%
\subsection{Properties of the model}\label{sec:analysis}

This section is dedicated to the investigation of the main properties of the kinetic epidemic model \fer{eq:kinetic_hesitancy}, where the Fokker--Planck operators $Q_J$ have been defined in \fer{op-FPH}. In the following, we will concentrate on the case where  $G\equiv 1$, so that the operators $Q_J$, $J \in \mathcal{C} = \{ S, E, I, R, V \}$, are given by
\begin{equations}\label{op-FP-simple}
Q_J(f_J, f_L)(w,t) = \frac{\sigma^2}{2}\frac{\partial^2 }{\partial w^2}\left((1-w^2)
 f_J(w,t)\right) + \frac{\partial }{\partial w}\left( \left( w -m(t)\right) f_J(w,t) \right),
 \end{equations}
where $m(t) = \lambda_J \mu(\rho_I) + \lambda_L m_L$. We complement the above equation with no-flux boundary conditions.  The local equilibrium distributions of \fer{op-FP-simple} for $J\in \mathcal C$ are now given by 
\[
f^{\mathrm{eq}}_J(w,t) = C_J(t) (1 - w)^{-1 + \frac{1 - m(t)}{\sigma^2}} (1 + w)^{-1 + \frac{1+ m(t)}{\sigma^2}},
\] 
with $C_J(t) > 0$ suitably chosen to normalize the distribution. This information is needed in particular to study the evolution of \eqref{eq:kinetic_hesitancy}.

%%%%%%%%%%%%%%%%%%%%%%%    NONNEGATIVITY    %%%%%%%%%%%%%%%%%%%%%%%%
\medskip
\noindent \textbf{Nonnegativity of solutions.} We begin by showing that solutions to \eqref{eq:kinetic_hesitancy} starting from a nonnegative initial datum remains nonnegative over time. To this end, we introduce a splitting strategy based on a the evaluation of the epidemic dynamics and on the opinion formation dynamics whose positivity preserving features are treated separately. We will follow ideas from Refs. \cite{CRS,ELL,FMZ}. In particular, in order to deal with $Q_J$ we need a preliminary result (which will be also useful to recover the uniqueness of solutions) concerning an $L^1$ comparison principle between solutions of the associated Fokker-Planck equations. It is the purpose of the following:

\begin{lemma} \label{lemma:L1-comparison}
Let $f_J$ be a solution of the Cauchy problem
\begin{equation} \label{eq:Cauchy-FP}
\begin{aligned}
& \frac{\partial}{\partial t} f_J(w,t) = \frac{\sigma^2}{2} \frac{\partial^2}{\partial w^2} \big( (1-w^2) f_J(w,t) \big) + \frac{\partial}{\partial w} \big( (w - m(t)) f_J(w,t) \big), \\[2mm]
& f_J(w,0) = f_J^{\mathrm{in}}(w),
\end{aligned}
\end{equation}
posed on $\mathcal{I} \times \R_+$, with the Fokker--Planck operator satisfying no-flux  boundary conditions. If $f_J^{\mathrm{in}} \in L^1(\mathcal{I})$, the $L^1$-norm of $f_J$ is non-increasing for $t \geq 0$.
\end{lemma}

\begin{proof}
Let us consider, for some parameter $\e > 0$, a regularized increasing approximation of the sign function $\sign_\e(x)$, $x \in \R$, and introduce a regularization $|f_J|_\e(w)$ of $|f_J|(w)$ as the primitive of $\sign_\e(f_J)(w)$, for any $w \in \mathcal{I}$. We can now multiply the Fokker-Planck equation by $\sign_\e(f_J)$ on both sides, integrate in $w$ and integrate by parts using no-flux boundary conditions to obtain
\begin{equation} \label{eq:initial-comparison}
\begin{aligned}
\frac{\dd}{\dd t} & \int_{-1}^1 |f_J|_\e (w,t) \dd w = - \int_{-1}^1 \sign_\e^\prime (f_J)(w,t) \big[ \partial_w f_J(w,t) \big]^2 D(w) \dd w \\[2mm]
& - \int_{-1}^1 \sign_\e^\prime(f_J)(w,t) f_J(w,t) \partial_w f_J(w,t) (w - m(t) + \partial_w D(w)) \dd w,
\end{aligned}
\end{equation}
where we have used the notation $D(w) = \frac{\sigma^2}{2}(1 - w^2) \geq 0$. Now, since
\begin{equation*}
\sign_\e^\prime(f_J) f_J \partial_w f_J = \partial_w \big( \sign_\e(f_J) f_J - |f_J|_\e \big),
\end{equation*}
one can use this information and further integrate by parts the second integral in \eqref{eq:initial-comparison}, to conclude that it vanishes in the limit $\e \to 0$ as $\sign_\e(f_J)(w) f_J(w) \to |f_J|_\e(w)$ for a.e. $w \in \mathcal{I}$. Because the integrand of the first term in \eqref{eq:initial-comparison} is nonnegative, we finally get the claim
\begin{equation*}
\frac{\dd}{\dd t} \|f_J(t)\|_{L^1(\mathcal{I})} \leq 0.
\end{equation*}

\end{proof}
From this lemma follows the nonnegativity of solutions for the Cauchy problem where only the Fokker--Planck operators are taken into account, as explained in the proof of the next result.

\begin{proposition}[Nonnegativity]
Let $(f_J)_{J \in \mathcal{C}}$ be a solution of the kinetic system \eqref{eq:kinetic_hesitancy} with $Q_J$ given by \eqref{op-FP-simple}, and let $f_J^{\mathrm{in}}(w) \in L^1(\mathcal{I})$ for any $J \in \mathcal{C}$. Then, for any $J \in \mathcal{C}$, if initially $f_J(w,0) = f_J^{\mathrm{in}}(w) \geq 0$ for any $w \in \mathcal{I}$, it holds that $f_J(w,t) \geq 0$ for a.e. $w \in \mathcal{I}$ and for any $t > 0$.
\end{proposition}

\begin{proof}
We start with the former ones and show that for any $J \in \mathcal{C}$ the solution of the Cauchy problem \eqref{eq:Cauchy-FP} remains nonnegative. This is done by defining the negative part of the solution $f_J$ via the regularization introduced above, as
\begin{equation*}
f_{J, \e}^{-} (w,t) = \frac{|f_J|_\e(w,t) - f_J(w,t)}{2}, \qquad w \in \mathcal{I}, \ t > 0,
\end{equation*}
and investigating its time evolution. Since the Fokker--Planck operators $Q_J$ are mass-preserving, thanks to Lemma \ref{lemma:L1-comparison} we deduce that, for any $t > 0$,
\begin{equation*}
\int_{-1}^1 f_{J, \e}^{-}(w,t) \dd w \leq \int_{-1}^1 f_{J, \e}^{\mathrm{in} -}(w) \dd w + \mathcal{O}(\e),
\end{equation*}
and the nonnegativity follows by taking the limit $\e \to 0$. To prove that also the remaining epidemic operators preserve the nonnegativity of solutions, we then study the ODE system
\begin{equation} \label{eq:Cauchy-ODE}
\begin{split}
\frac{\partial f_S(w,t)}{\partial t} &= - f_S(w,t) \left(\kappa_V(w) +  \int_{-1}^1  \kappa_T(w,w_*) f_I(w_*,t) \dd w_*\right)  , \\[2mm]
\frac{\partial f_E(w,t)}{\partial t} &= f_S(w,t) \int_{-1}^1 \kappa_T(w,w_*) f_I(w_*,t) \dd w_* - \sigma_E f_E(w, t), \\[2mm]
\frac{\partial f_I(w,t)}{\partial t} \ &= \sigma_E f_E(w,t) - \gamma_I f_I(w,t), \\[4mm]
\frac{\partial f_R(w,t)}{\partial t} &= \gamma_I f_I(w,t), \\[4mm]
\frac{\partial f_V(w,t)}{\partial t} &= f_S(w,t) \kappa_V(w)
\end{split}
\end{equation}
where we recall that $\kappa_T, \kappa_V \geq 0$ for any $w, w_* \in \mathcal{I}$. The idea is to proceed by contradiction, supposing that there exist a time $t_0 > 0$ and an opinion $w_0 = w_0(t_0) \in \mathcal{I}$ such that, for any $t \in [0, t_0)$ and any $w \in \mathcal{I}$,
\begin{equation*}
f_E(w, t) > 0, \qquad f_E(w_0, t_0) = 0, \qquad \frac{\partial f_E(w_0, t_0)}{\partial t} < 0.
\end{equation*}
But now we can use the first and third equations of system \eqref{eq:Cauchy-ODE} to get explicit expressions for $f_S$ and $f_I$, which write, for any $t \in [0, t_0]$ and any $w \in \mathcal{I}$,
\begin{equation} \label{eq:integration}
\begin{aligned}
f_S(w, t) & = f_S^{\mathrm{in}}(w) e^{- \int_0^t \left( \int_{-1}^1   \kappa_T(w,w_*) f_I(w_*, s) \dd w_* + \kappa_V(w) \right) \dd s} \geq 0, \\ 
f_I(w, t) & = f_I^{\mathrm{in}}(w) e^{- \gamma_I t} + \sigma_E \int_0^t f_E(w, s) e^{- \gamma_I (t - s)} \dd s \geq 0,
\end{aligned}
\end{equation}
since, by hypothesis, $f^{\mathrm{in}}_S(w),\ f_I^{\mathrm{in}}(w) \geq 0$ for $w \in \mathcal{I}$. This implies in particular, from the second equation in \eqref{eq:Cauchy-ODE}, that
\begin{equation*}
\frac{\partial f_E(w_0, t_0)}{\partial t} = f_S(w_0, t_0) \int_{-1}^1 \kappa_T(w, w_*) f_I(w_0, t_0) \dd w_* \geq 0,
\end{equation*}
which contradicts our hypothesis, so that indeed $f_E(t, w) \geq 0$ for any $t > 0$ and any $w \in \mathcal{I}$. Then also $f_J(t, w) \geq 0$ for any $J \in \{ S, I, R, V \}$, $t > 0$ and any $w \in \mathcal{I}$, thanks again to \eqref{eq:integration} and by also integrating the evolution equations for $f_R$ and $f_V$ in \eqref{eq:Cauchy-ODE}. This ends our proof.

\end{proof}

%%%%%%%%%%%%%%%%%%%%%%%%    REGULARITY    %%%%%%%%%%%%%%%%%%%%%%%%%%
\noindent \textbf{Regularity of solutions.} Another important aspect is to understand whether the solutions to our model preserve some sort of regularity of the initial conditions. We first observe that the $L^1$ regularity easily follows from the nonnegativity result just proved in the previous part and by the conservation of mass which is intrinsic to the model. Indeed, let us suppose that $f_J^{\mathrm{in}} \in L^1(\mathcal{I})$ for any $J \in \mathcal{C}$. Since for any $J \in \mathcal{C}$, $f_J(w,t) \geq 0$ for a.e. $w \in \mathcal{I}$ and any $t > 0$, we deduce that $\| f_J \|_{L^1(\mathcal{I})} = \int_{-1}^1 f_J(w,t) \dd w$ for any $t > 0$. Then, by summing the four equations of the kinetic system \eqref{eq:kinetic_hesitancy} with $Q_J$ given by \eqref{op-FP-simple}, and by integrating in $w$ over $\mathcal{I}$, we infer that for any $t > 0$
\begin{equation*}
\sum_{J \in \mathcal{C}} \| f_J \|_{L^1(\mathcal{I})} = \sum_{J \in \mathcal{C}} \| f_J^{\mathrm{in}} \|_{L^1(\mathcal{I})},
\end{equation*}
since the Fokker--Planck operators $Q_J$ are mass-preserving. Therefore, $f_J \in L^1(\mathcal{I})$, $J \in \mathcal{C}$, for any $t > 0$.

In order to obtain more regularity, we essentially need to ensure that we can perform suitable integrations by parts when dealing with the operators $Q_J$, while it is clear that the operators $K_T$, $K_V$ can be dealt with easily thanks to the $L^1$ regularity. In particular, the main issue comes from the fact that the diffusion part of the Fokker--Planck operators does not have a self-adjoint structure. In order to solve the problem, we then rewrite $Q_J$, for any $J \in \mathcal{C}$, as
\begin{equation*}
Q_J(f_J,f_L)(w,t) = \frac{\sigma^2}{2} \frac{\partial}{\partial w} \left( (1 - w^2) \frac{\partial}{\partial w} f_J(w,t) \right) + \frac{\partial}{\partial w} \Big( \big((1- \sigma^2) w - m(t) \big) f_J(w,t) \Big),
\end{equation*}
with an augmented drift term, and now the diffusion part has a more manageable form. However, this choice comes with the tradeoff that additional boundary conditions need to be imposed, as shown in the following result.

\begin{proposition}[Regularity]
Let $p \in [2, +\infty)$ and consider a solution $(f_J)_{J \in \mathcal{C}}$ of the kinetic system \eqref{eq:kinetic_hesitancy} where the operators $Q_J$ are given by \eqref{op-FP-simple}, while the operators $K_T$ and $K_V$ are defined by \fer{eq:K_def} and \fer{eq:KVdef}  respectively, with $\kappa_T \in L^\infty(\mathcal{I} \times \mathcal{I})$ and $\kappa_V \in L^\infty(\mathcal{I})$. We further assume that, for any $J \in \mathcal{C}$, $Q_J$ satisfy the additional no-flux boundary conditions
\begin{equation} \label{regularity-BC}
\begin{aligned}
\big((1- \sigma^2) w - m(t) \big) f_J(w,t) \Big|_{w = \pm 1} = 0, & \\[4mm]
(1-w^2) \frac{\partial}{\partial w} f_J(w,t) \Big|_{w = \pm 1} = 0. &
\end{aligned}
\end{equation}
for any $t > 0$. Then, if initially $f_J^{\mathrm{in}} \in L^p(\mathcal{I})$ for any $J \in \mathcal{C}$, the regularity is preserved over time, i.e. $f_J \in L^p(\mathcal{I})$, $J \in \mathcal{C}$, for any $t > 0$.
\end{proposition}

\begin{proof}
Fix $p \geq 2$ and multiply each equation in the kinetic system \eqref{eq:kinetic_opinion} by $f_J^{p-1}$, with respect to the corresponding index $J \in \mathcal{C}$. After integration in $w$, we initially obtain
\begin{equation} \label{eq:initial-regularity}
\begin{aligned}
\frac{1}{p}\frac{\dd}{\dd t} \| f_S \|_{L^p(\mathcal{I})}^p & = - \int_{-1}^1 \Big( K_T(f_S, f_I) + K_V(f_S) \Big) f_S^{p-1} \dd w + \frac{1}{\tau} \int_{-1}^1 Q_S(f_S,f_L) f_S^{p-1} \dd w, \\[2mm]
\frac{1}{p}\frac{\dd}{\dd t} \| f_E \|_{L^p(\mathcal{I})}^p & = \int_{-1}^1 K_T(f_S, f_I) f_E^{p-1} \dd w - \sigma_E \| f_E \|_{L^p(\mathcal{I})}^p + \frac{1}{\tau} \int_{-1}^1 Q_E(f_E,f_L) f_E^{p-1} \dd w, \\[2mm]
\frac{1}{p}\frac{\dd}{\dd t} \| f_I \|_{L^p(\mathcal{I})}^p & =  \sigma_E \int_{-1}^1 f_E f_I^{p-1} \dd w - \gamma_I \| f_I \|_{L^p(\mathcal{I})}^p + \frac{1}{\tau} \int_{-1}^1 Q_I(f_I,f_L) f_I^{p-1} \dd w, \\[2mm]
\frac{1}{p}\frac{\dd}{\dd t} \| f_R \|_{L^p(\mathcal{I})}^p & =  \gamma_I \int_{-1}^1 f_I f_R^{p-1} \dd w + \frac{1}{\tau} \int_{-1}^1 Q_R(f_R,f_L) f_R^{p-1} \dd w, \\[2mm]
\frac{1}{p}\frac{\dd}{\dd t} \| f_V \|_{L^p(\mathcal{I})}^p & = \int_{-1}^1 K_V(f_S) f_V^{p-1} \dd w + \frac{1}{\tau} \int_{-1}^1 Q_V(f_V,f_L) f_V^{p-1} \dd w
\end{aligned}
\end{equation}
Let us deal initially with the terms involving $K_T$ and $K_V$. For those appearing in the first equation, we easily get
\begin{equation*}
\begin{aligned}
-\int_{-1}^1 \Big( K_T(f_S, f_I) + & K_V(f_S) \Big) f_S^{p-1} \dd w \\[2mm]
& \leq \Big( \| \kappa_T \|_{L^\infty(\mathcal{I} \times \mathcal{I})} \| f_I \|_{L^1(\mathcal{I})} + \| \kappa_V \|_{L^\infty(\mathcal{I})} \Big) \| f_S \|_{L^p(\mathcal{I})}^p,
\end{aligned}
\end{equation*}
while the similar terms in the evolution equations for $f_E$ and $f_V$ can be treated by successively applying Hölder's and Young's inequalities with exponents $p$ and $\frac{p}{p-1}$, recovering
\begin{equation*}
\begin{aligned}
& \int_{-1}^1 K_T(f_S, f_I) f_E^{p-1} \dd w \leq \| \kappa_T \|_{L^\infty(\mathcal{I} \times \mathcal{I})} \| f_I \|_{L^1(\mathcal{I})} \left( \frac{1}{p} \| f_S \|_{L^p(\mathcal{I})}^p + \frac{p-1}{p} \| f_E \|_{L^p(\mathcal{I})}^p \right), \\[2mm]
& \int_{-1}^1 K_V(f_S) f_V^{p-1} \dd w \leq \| \kappa_V \|_{L^\infty(\mathcal{I})} \left( \frac{1}{p} \| f_S \|_{L^p(\mathcal{I})}^p + \frac{p-1}{p} \| f_V \|_{L^p(\mathcal{I})}^p \right).
\end{aligned}
\end{equation*}
A similar argument can be used to treat the mixed product terms appearing in the evolution equations for the distributions $f_I$ and $f_R$.

At last, we study the four terms involving the Fokker--Planck operators $Q_J$. For any $J \in \mathcal{C}$, these terms explicitly write
\begin{equation*}
\begin{aligned}
& \int_{-1}^1 Q_J(f_J,f_L) f_J^{p-1} \dd w \\[2mm]
& = \int_{-1}^1 f_J^{p-1} \left( \frac{\sigma}{2} \frac{\partial^2}{\partial w^2} \big( (1 - w^2) f_J \big) + \frac{\partial}{\partial w} \big( (w - m(t)) f_J \big) \right) \dd w \\[2mm]
&  = \frac{\sigma^2}{2} \int_{-1}^1 f_J^{p-1} \frac{\partial}{\partial w} \left( (1 - w^2) \frac{\partial}{\partial w} f_J \right) \dd w + \int_{-1}^1 f_J^{p-1} \frac{\partial}{\partial w} \Big( \big( (1 - \sigma^2) w - m(t) \big) f_J \Big) \dd w \\[4mm]
& =: \mathcal{T}_1 + \mathcal{T}_2,
\end{aligned}
\end{equation*}
by performing one derivation in the diffusion term. Now, integrating by parts $\mathcal{T}_1$ and using the second of the boundary conditions \eqref{regularity-BC}, allow to show that it is nonpositive
\begin{equation*}
\begin{aligned}
\mathcal{T}_1 & = - \frac{\sigma^2}{2} \int_{-1}^1 (1 - w^2) \frac{\partial}{\partial w} f_J^{p-1} \frac{\partial}{\partial w} f_J \dd w \\[2mm]
& = - \frac{(p-1) \sigma^2}{2} \int_{-1}^1 (1 - w^2) \left( \frac{\partial}{\partial w} f_J \right)^2 f_J^{p-2} \dd w \leq 0
\end{aligned}
\end{equation*}
since $f_J$ is nonnegative, and one can thus get rid of this term in the estimate. In order to deal with $\mathcal{T}_2$, instead, the idea is to exploit its antisymmetric nature. For this, we can on one side perform an integration by parts and use the first one of the boundary conditions \eqref{regularity-BC} to get
\begin{equation*}
\begin{aligned}
\mathcal{T}_2 & = - \int_{-1}^1 \frac{\partial}{\partial w} f_J^{p-1} \big( (1 - \sigma^2) w - m(t) \big) f_J \dd w, \\[2mm]
& = - (p-1) \int_{-1}^1 \big( (1 - \sigma^2) w - m(t) \big) f_J^{p-1} \frac{\partial}{\partial w} f_J \dd w.
\end{aligned}
\end{equation*}
On the other hand, by computing the derivative with respect to $w$ in $\mathcal{T}_2$, one also obtains
\begin{equation*}
\mathcal{T}_2 = (1 - \sigma^2) \int_{-1}^1 f_J^{p} \dd w + \int_{-1}^1 \big( (1 - \sigma^2) w - m(t) \big) f_J^{p-1} \frac{\partial}{\partial w} f_J \dd w.
\end{equation*}
Therefore, one can split $\mathcal{T}_2$ as $\mathcal{T}_2 = \frac{1}{p} \mathcal{T}_2 + \frac{p-1}{p} \mathcal{T}_2$ and use, as shown, the integration by parts on the first addend and the direct derivation in the second one, to get rid of the integral $ \int_{-1}^1 \big( (1 - \sigma^2) w - m(t) \big) f_J^{p-1} \frac{\partial}{\partial w} f_J \dd w$. In particular, one ends up with the final estimate
\begin{equation*}
\int_{-1}^1 Q_J(f_J,f_L) f_J^{p-1} \dd w \leq |1-\sigma^2| \| f_J \|_{L^p(\mathcal{I})}^p,
\end{equation*}
for any $J \in \mathcal{C}$.

Collecting all the bounds obtained for the different terms appearing in \eqref{eq:initial-regularity}, we finally infer that
\begin{equation*}
\frac{\dd}{\dd t} \sum_{J \in \mathcal{C}} \| f_J \|_{L^p(\mathcal{I})}^p \leq C \sum_{J \in \mathcal{C}}  \| f_J \|_{L^p(\mathcal{I})}^p,
\end{equation*}
with a constant $C = C\big( p, \sigma_E, \gamma_I, |1-\sigma^2|, \| \kappa_T \|_{L^\infty(\mathcal{I}\times \mathcal{I})}, \| \kappa_V \|_{L^\infty(\mathcal{I})}, \| f_I \|_{L^1(\mathcal{I})} \big)$, and we conclude on the regularity estimate by using Grönwall's Lemma.

\end{proof}

%%%%%%%%%%%%%%%%%%%%%%%%%    UNIQUENESS    %%%%%%%%%%%%%%%%%%%%%%%%%
\noindent \textbf{Uniqueness of solutions.} We then proceed by showing a uniqueness result, in the particular case where the two solutions possess the same mass fraction in the compartment of infected (and thus, in all compartments). This assumption is made for technical reasons, in order to linearize the Fokker--Planck operators and allow for a comparison between solutions using Lemma \ref{lemma:L1-comparison}.

\begin{proposition}[Uniqueness] 
Let $(f_J)_{J \in \mathcal{C}}$ and $(g_J)_{J \in \mathcal{C}}$ be two solutions of the kinetic system \eqref{eq:kinetic_hesitancy} where $Q_J$ is given by \eqref{op-FP-simple} and $K_T$, $K_V$ are defined by \fer{eq:K_def} and \fer{eq:KVdef} respectively, with $\kappa_T \in L^\infty(\mathcal{I} \times \mathcal{I})$ and $\kappa_V \in L^\infty(\mathcal{I})$. Suppose that $f_I$ and $g_I$ satisfy
\begin{equation*}
\int_{-1}^1 f_I(w, t) \dd w = \int_{-1}^1 g_I(w, t) \dd w \qquad \forall t > 0,
\end{equation*}
and let $f_J(w, 0) = f_J^{\mathrm{in}}(w) \in L^1(\mathcal{I})$ and $g_J(w, 0) = g_J^{\mathrm{in}}(w) \in L^1(\mathcal{I})$, for any $J \in \mathcal{C}$. If  $f_J^{\mathrm{in}}(w) = g_J^{\mathrm{in}}(w)$ for a.e. $w \in \mathcal{I}$, then $f_J(w,t) = g_J(w,t)$ for a.e. $w \in \mathcal{I}$ and for any $t > 0$.
\end{proposition}

\begin{proof}
We shall prove $L^1$-uniqueness of the solutions. Let us define the mass fractions of infected for the two solutions as $\rho_I^f(t) = \int_{-1}^1 f_I(w, t) \dd w$ and $\rho_I^g(t) = \int_{-1}^1 g_I(w, t) \dd w$. Since by assumption $\rho_I^f(t) = \rho_I^g(t)$ for any $t > 0$, one deduces that also $\mu_I^f(t) = \mu_I^g(t)$ for any $t > 0$, where we have used the notations $\mu_I^f(t) = \mu(\rho_I^f(t))$ and $\mu_I^f(t) = \mu(\rho_I^f(t))$. This is equivalent to linearize the drift term in the Fokker--Planck operators, so that the differences $f_J(w,t) - g_J(w,t)$, $J \in \mathcal{C}$, are solutions over $\mathcal{I} \times \R_+$ of the system
\begin{equation} \label{eq:initial-uniqueness}
\begin{split}
\frac{\partial (f_S - g_S)}{\partial t} & = - \big( K_T(f_S,f_I) - K_T(g_S, g_I) \big) - K_V(f_S - g_S) + \dfrac{1}{\tau} Q_{S}(f_S - g_S), \\
\frac{\partial (f_E - g_E)}{\partial t} & = K_T(f_S,f_I) - K_T(g_S, g_I) - \sigma_E (f_E - g_E) + \dfrac{1}{\tau}Q_{E}(f_E - g_E), \\
\frac{\partial (f_I - g_I)}{\partial t} \ & = \sigma_E (f_E - g_E) - \gamma_I (f_I - g_I) + \dfrac{1}{\tau} Q_{I}(f_I - g_I), \\
\frac{\partial (f_R - g_R)}{\partial t} & = \gamma_I (f_I - g_I) + \dfrac{1}{\tau}Q_{R}(f_R - g_R), \\[2mm]
\frac{\partial (f_V - g_V)}{\partial t} & = K_V(f_S - g_S) + \dfrac{1}{\tau}Q_{V}(f_V - g_V).
\end{split}
\end{equation}
Notice that we have dropped the dependencies of the operators $Q_J$ on $f_L$ for a sake of clarity, since it is obvious that the static distribution of leaders is the same for the evolution systems of $(f_J)_{J \in \mathcal{C}}$ and $(g_J)_{J \in \mathcal{C}}$.

Now, the idea is to recover an $L^1$ comparison principle for the differences $f_J - g_J$ to get rid of the Fokker--Planck operators. This is done by introducing an increasing regularized approximation of the sign function $\sign_\e(x)$, $x \in \R$, and define through it a regularization $|f_J - g_J|_\e$ of $|f_J - g_J|$ for any $J \in \mathcal{C}$. Multiplying each equation in \eqref{eq:initial-uniqueness} by the corresponding $\sign_\e(f_J - g_J)$ and integrating in $w$, the same steps performed in the proof of Lemma \ref{lemma:L1-comparison} lead, by taking the limit $\e \to 0$, to the initial estimate
\begin{equation*}
\begin{split}
\frac{\dd}{\dd t} \| f_S - g_S \|_{L^1(\mathcal{I})} \ & \leq  \int_{-1}^1 \big| K_T(f_S,f_I) - K_T(g_S, g_I) \big| \dd w + \int_{-1}^1 \big| K_V(f_S - g_S) \big| \dd w, \\[2mm]
\frac{\dd}{\dd t} \| f_E - g_E \|_{L^1(\mathcal{I})} & \leq  \int_{-1}^1 \big| K_T(f_S,f_I) - K_T(g_S, g_I) \big| \dd w + \sigma_E \| f_E - g_E \|_{L^1(\mathcal{I})}, \\[2mm]
\frac{\dd}{\dd t} \| f_I - g_I \|_{L^1(\mathcal{I})} & \leq \sigma_E \| f_E - g_E \|_{L^1(\mathcal{I})} + \gamma_I \| f_I - g_I \|_{L^1(\mathcal{I})}, \\[4mm]
\frac{\dd}{\dd t} \| f_R - g_R \|_{L^1(\mathcal{I})} & \leq \gamma_I \| f_I - g_I \|_{L^1(\mathcal{I})}, \\[2mm]
\frac{\dd}{\dd t} \| f_V - g_V \|_{L^1(\mathcal{I})} & \leq \int_{-1}^1 \big| K_V(f_S - g_S) \big| \dd w.
\end{split}
\end{equation*}
To conclude, it only remains to control the integral terms involving the operators $K_T$ and $K_V$. Since $\kappa_T \in L^\infty(\mathcal{I} \times \mathcal{I})$ and $\kappa_V \in L^\infty(\mathcal{I})$ by assumption, one easily obtains the estimate
\begin{equation*}
\begin{aligned}
\int_{-1}^1 \big| K(f_S,f_I) - & K(g_S, g_I) \big| \dd w \\[2mm]
& \leq \| \kappa_T \|_{L^\infty(\mathcal{I} \times \mathcal{I})} \Big( \| f_I \|_{L^1(\mathcal{I})} \| f_S - g_S \|_{L^1(\mathcal{I})} + \| g_S \|_{L^1(\mathcal{I})} \| f_I - g_I \|_{L^1(\mathcal{I})} \Big),
\end{aligned}
\end{equation*}
and the obvious bound $\int_{-1}^1 \big| K_V(f_S - g_S) \big| \dd w \leq \| \kappa_V \|_{L^\infty(\mathcal{I})} \| f_S - g_S \|_{L^1(\mathcal{I})} $. Injecting these into the bound on the full system, we then conclude that
\begin{equation*}
\frac{\dd}{\dd t}\sum_{J \in \mathcal{C}} \| f_J - g_J \|_{L^1(\mathcal{I})} \leq C \sum_{J \in \mathcal{C}} \| f_J - g_J \|_{L^1(\mathcal{I})},
\end{equation*}
where $C = C\big( \sigma_E, \gamma_I, \| \kappa_T \|_{L^\infty(\mathcal{I}\times \mathcal{I})}, \| \kappa_V \|_{L^\infty(\mathcal{I})}, \| f_I \|_{L^1(\mathcal{I})}, \| g_S \|_{L^1(\mathcal{I})} \big)$, and an application of Grönwall's Lemma allows to conclude on the uniqueness, since by hypothesis $f_J$ and $g_J$ are equal initially.

\end{proof}

%%%%%%%%%%%%%%%%%%%%%%%%%%%%%%%%%%%%%%%%%%%%%%%%%%%%%%%%%%%%
%%%%%%%%%%%%%%%%%%%%%%%%%%%%%%%%%%%%%%%%%%%%%%%%%%%%%%%%%%%%
%%%%%%%%%%%%%%%%%%    SECTION 4: NUMERICAL EXPERIMENTS    %%%%%%%%%%%%%%%%%%
%%%%%%%%%%%%%%%%%%%%%%%%%%%%%%%%%%%%%%%%%%%%%%%%%%%%%%%%%%%%
%%%%%%%%%%%%%%%%%%%%%%%%%%%%%%%%%%%%%%%%%%%%%%%%%%%%%%%%%%%%
\section{Numerical results}\label{sec:num}

We conclude this work by presenting a series of numerical examples that highlight the main features of the SEIRV kinetic model \fer{eq:kinetic_hesitancy}. We focus our numerical investigation to the case where $K_T(\cdot,\cdot)$ and $K_V(\cdot)$ have the form
\begin{equation*}
\begin{aligned}
K_T(f_S,f_I)(w,t) & =  \frac{\beta_T}{4} (1 - w) f_S(w,t) (1 - m_I(t)) \rho_I(t), \\[2mm]
K_V(f_S)(w,t) & =  \frac{\beta_V}{2} (1 + w) f_S(w,t),
\end{aligned}
\end{equation*}
corresponding to $\kappa_T(w,w_*) = \frac{\beta_T}{4}(1-w)(1-w_*)$, obtained from \fer{eq:kappa_def} with $\alpha_T = 1$, and $\kappa_V(w) = \frac{\beta_V}{2}(1+w)$, obtained from \fer{eq:kappav_def} with $\alpha_V = 1$, for any $w \in \mathcal{I}$. 

The approximation of the system of kinetic equations is obtained through a second-order splitting strategy to deal separately with the Fokker--Planck opinion formation process and wtih the epidemiological operators. To this end, we introduce a discretization of any time interval $[0,T]$ with a time step $\Delta t > 0$ such that $ t^n = n\Delta t$, and we denote by $f_J^n$ the approximation of $f_J(w,t^n)$, $J \in \mathcal{C}$, for any $w \in \mathcal{I}$ and any $0 \leq n \leq N_t$ with $ N_t \Delta t = T$. We introduce a splitting strategy between the opinion formation step $f_J^* = \mathcal O_{\Delta t}(f_J^n)$
\begin{equation*}
\begin{cases}
&\frac{\partial}{\partial t} f_J^*(w) = \frac{1}{\tau} Q_J(f_J^*(w),f_L(w)), \\
&f_J^*(w,0) = f_J^n(w), \quad J \in \mathcal{C},
\end{cases}
\end{equation*}
and the epidemiological step $f^{**}_J =\mathcal E_{\Delta t}(f_J^{**})$
\begin{equation*}
\begin{cases}
\frac{\partial}{\partial t} f_S^{**}(w) & = - \frac{\beta_T}{4} (1-w) f_S^{**}(w) (1 - m_I^{**}) \rho_I^{**} - \frac{\beta_V}{2} (1+w) f_S^{**}(w), \\[4mm]
\frac{\partial}{\partial t} f_E^{**}(w) & = \frac{\beta_T}{4} (1-w) f_S^{**}(w) (1 - m_I^{**}) \rho_I^{**} - \sigma_E f_E^{**}(w), \\[4mm]
\frac{\partial}{\partial t} f_I^{**}(w) & = \sigma_E f_E^{**}(w) - \gamma_I f_I^{**}(w), \\[4mm]
\frac{\partial}{\partial t} f_R^{**}(w) & = \gamma_I f_I^{**}(w), \\[4mm]
\frac{\partial}{\partial t} f_V^{**}(w) & = \frac{\beta_V}{2} (1+w) f_S^{**}(w), \\[4mm]
f_J^{**}(w,0) & = f_J^*(w,\Delta t), \quad J \in \mathcal{C},. 
\end{cases}
\end{equation*}
The operators $Q_J(\cdot,\cdot)$ have been defined in \fer{op-FPH} and are complemented by no-flux boundary conditions. Hence, the approximated solution at time $t^{n+1}$ is given by the combination of the two introduced steps, i.e. if a first order splitting strategy is considered we have
\[
f^{n+1}_J(w) = \mathcal E_{\Delta t}(\mathcal O_{\Delta t}(f_J^n(w))),
\]
whereas the second-order Strang splitting method is obtained as
\[
f^{n+1}_J(w) = \mathcal E_{\Delta t/2}(\mathcal O_{\Delta t}(\mathcal E_{\Delta t/2}(f_J^n(w)))), 
\]
for any $J\in\mathcal C$. In particular we will solve the opinion consensus step with a  semi-implicit structure preserving scheme (SP) for Fokker--Planck-type equations, see Ref. \cite{PZ}. On the other hand, the numerical integration of the epidemiological step is obtained with a RK4 method. 
 
The main advantage of this approach is its ability to reproduce large time statistical properties of the exact steady state of the full system \eqref{eq:kinetic_hesitancy} with arbitrary accuracy, while also preserving the main physical properties of the solution, like positivity and entropy dissipation.
In all examples, we consider the epidemiological parameters enlisted in Table \ref{tab:1}.  These values are coherent with the existing literature for COVID-19 pandemics, see Ref. \cite{BDM,Gatto,Z}. Moreover, inside the drift terms of the Fokker--Planck operators we assume for simplicity that $\lambda_J = \lambda$ for any $J \in \mathcal{C}$. 

\begin{table}
\begin{center}
\begin{tabular}{| c | c | c |}
\hline
\cellcolor{lightgray}{Parameter} & \cellcolor{lightgray}{Meaning} & \cellcolor{lightgray}{Value} \\
$\beta_T$ & baseline contact  & 0.4 \\
$\beta_V$ & baseline vaccination & $5\cdot 10^{-3}$ \\
$\gamma_I$ & recovery rate & 1/7 \\
$\sigma_E$ & latency & 1/2 \\
$\sigma^2$ & diffusion constant & $10^{-2} $ \\
\hline
\end{tabular}
\label{tab:1}
\end{center}
\caption{Epidemiological parameters. }
\end{table}

%%%%%%%%%%%%%%%%%%%%%%    SUBSECTION: TEST 1    %%%%%%%%%%%%%%%%%%%%%%%%
\subsection{Test 1: Influence of opinion formation in epidemic dynamics}

In the first example, we study the evolution of the kinetic system \eqref{eq:kinetic_hesitancy} when only the influence of the function $\mu(\rho_I)$ is taken into account through the operators \eqref{op-FPH}. For this, we fix the compromise propensity $\lambda = 1$, which gives $\lambda_L = 0$ from the relation $\lambda + \lambda_L = 1$, meaning that we exclude any external effect of the leader--follower interactions. Furthermore, the opinion variation depends on the quantity defined in \fer{risk}, for which we consider the two values $n=4$ and $n=20$. In order to highlight the outcome of this test case, we start from a favorable scenario where the individuals have only positive opinions, which corresponds to the choices 
\begin{equation}
\label{eq:fJ0_t1}
f_J(w,0) = \rho_J(0) \chi_{[0,1]}(w), \qquad J \in \mathcal{C},
\end{equation}
for the initial conditions of the compartmental distributions. Furthermore, we have considered the initial masses $\rho_J(0) = 10^{-3}$, $J \in \{ E, I, R \}$, $\rho_V(0) = 0$ and $\rho_S(0) = 1 - \rho_E(0) - \rho_I(0) - \rho_R(0)$.

\begin{figure}
\centering
\includegraphics[scale = 0.31]{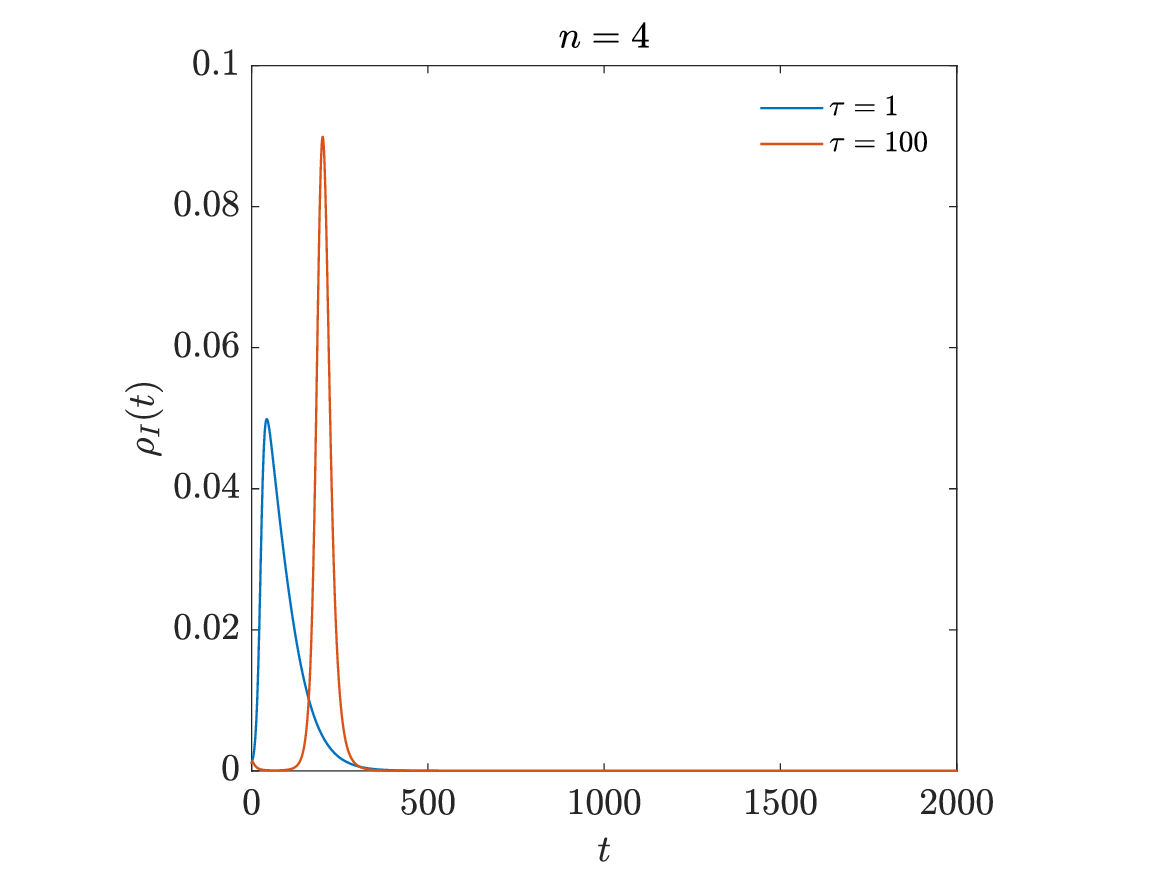} 
\includegraphics[scale = 0.31]{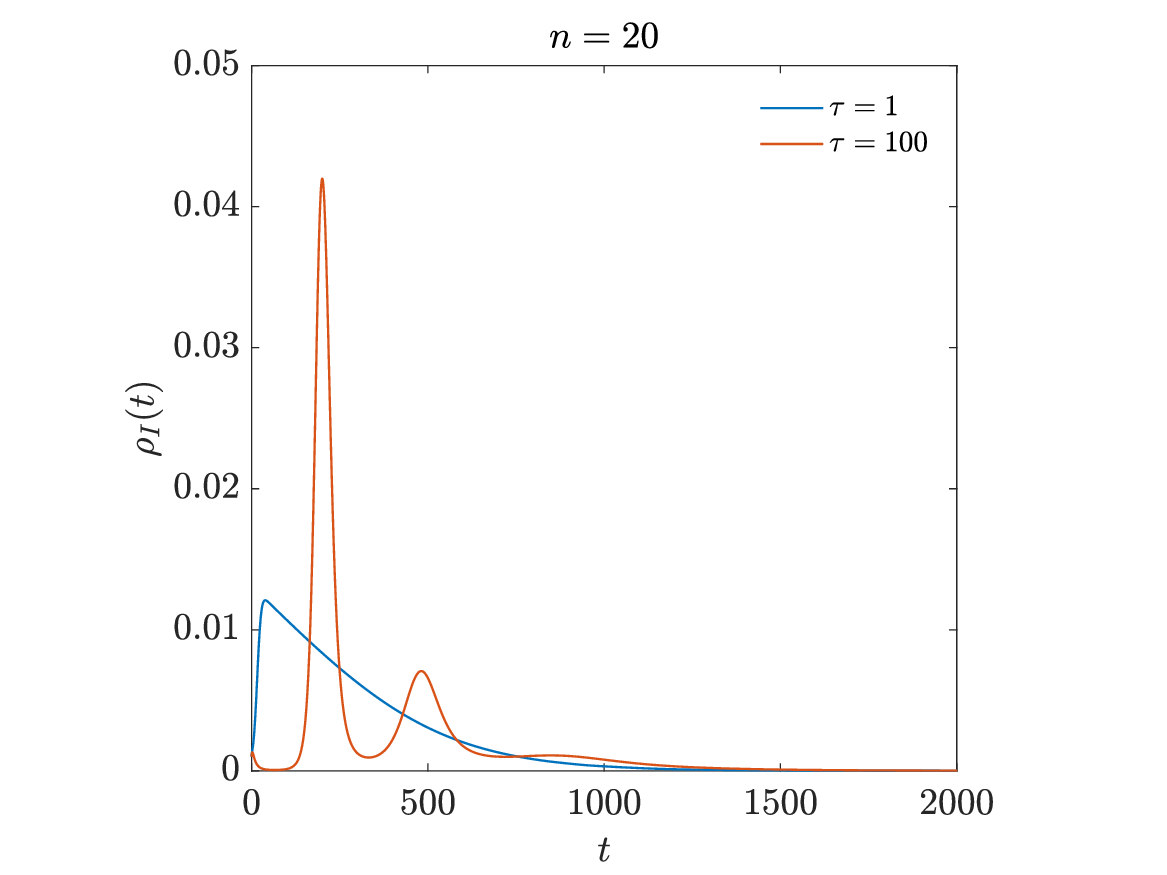} \\\includegraphics[scale = 0.31]{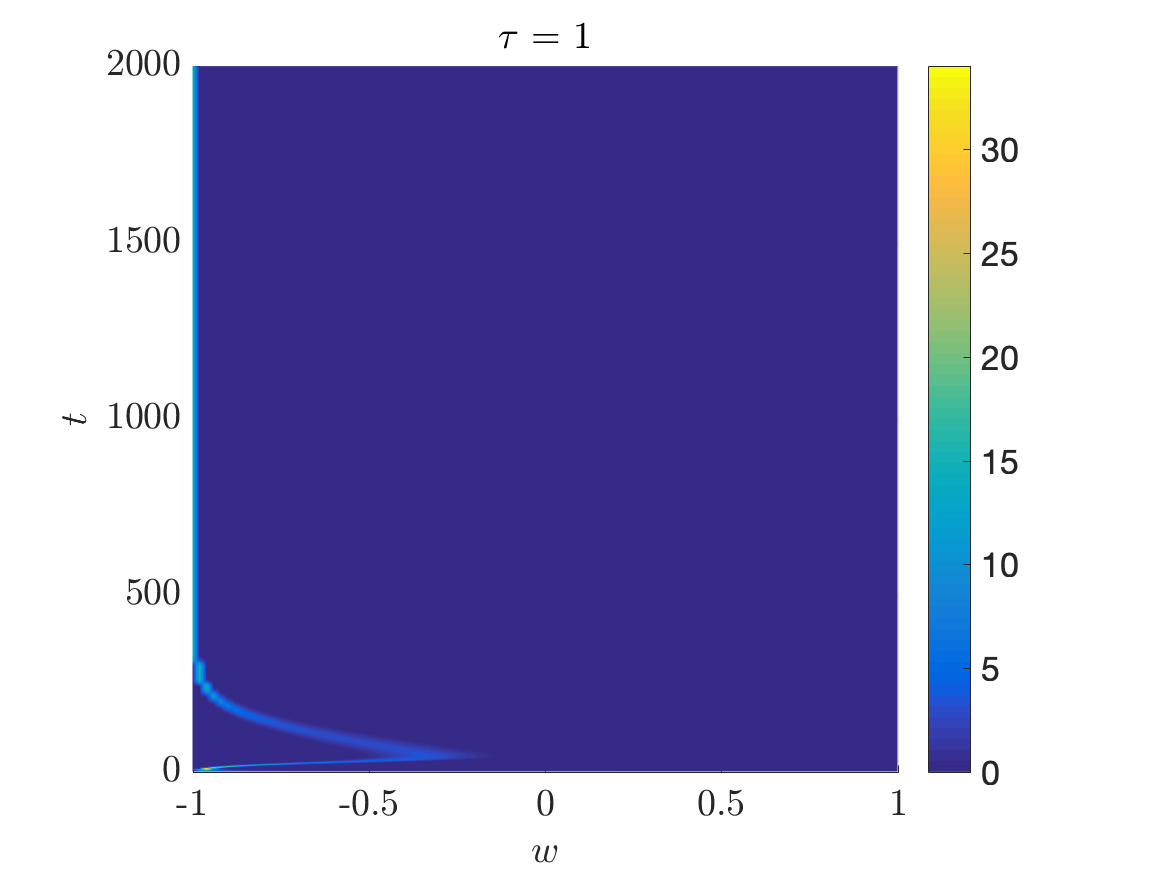} 
\includegraphics[scale = 0.31]{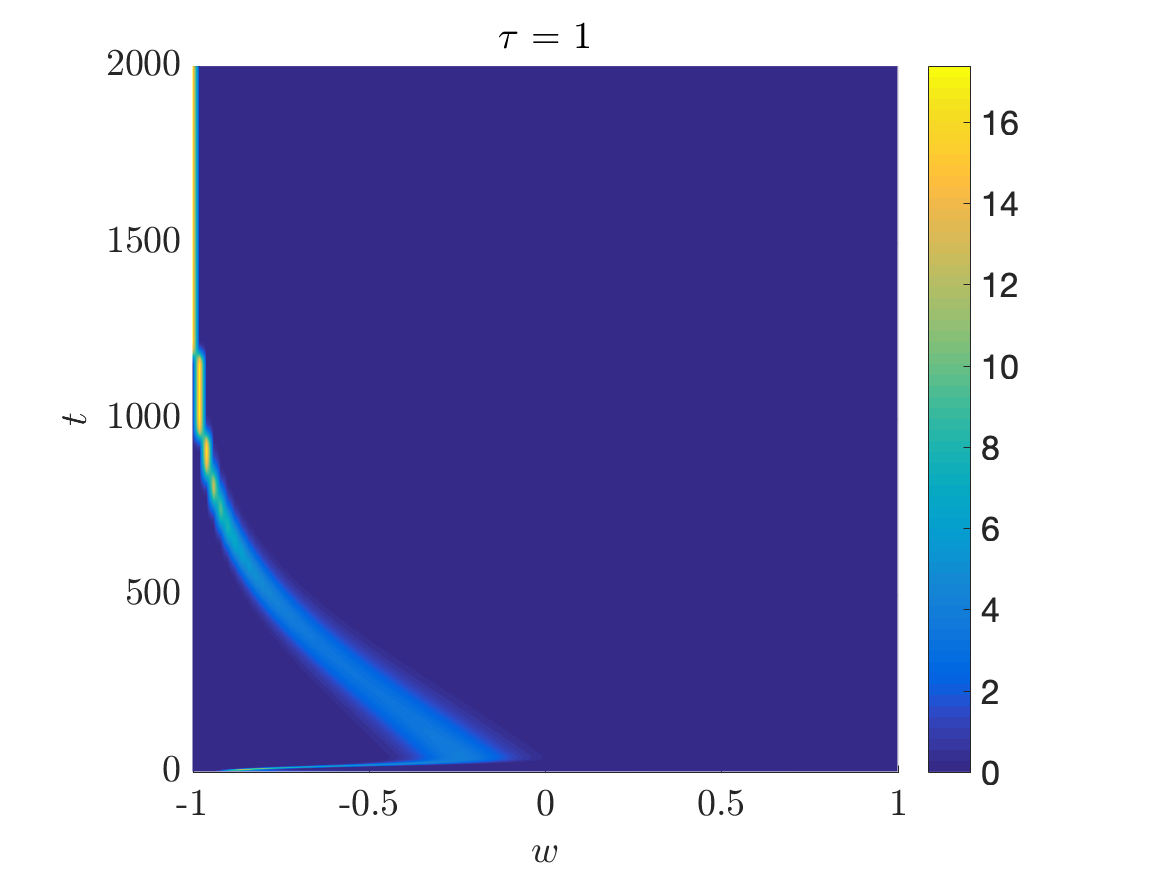} \\\includegraphics[scale = 0.31]{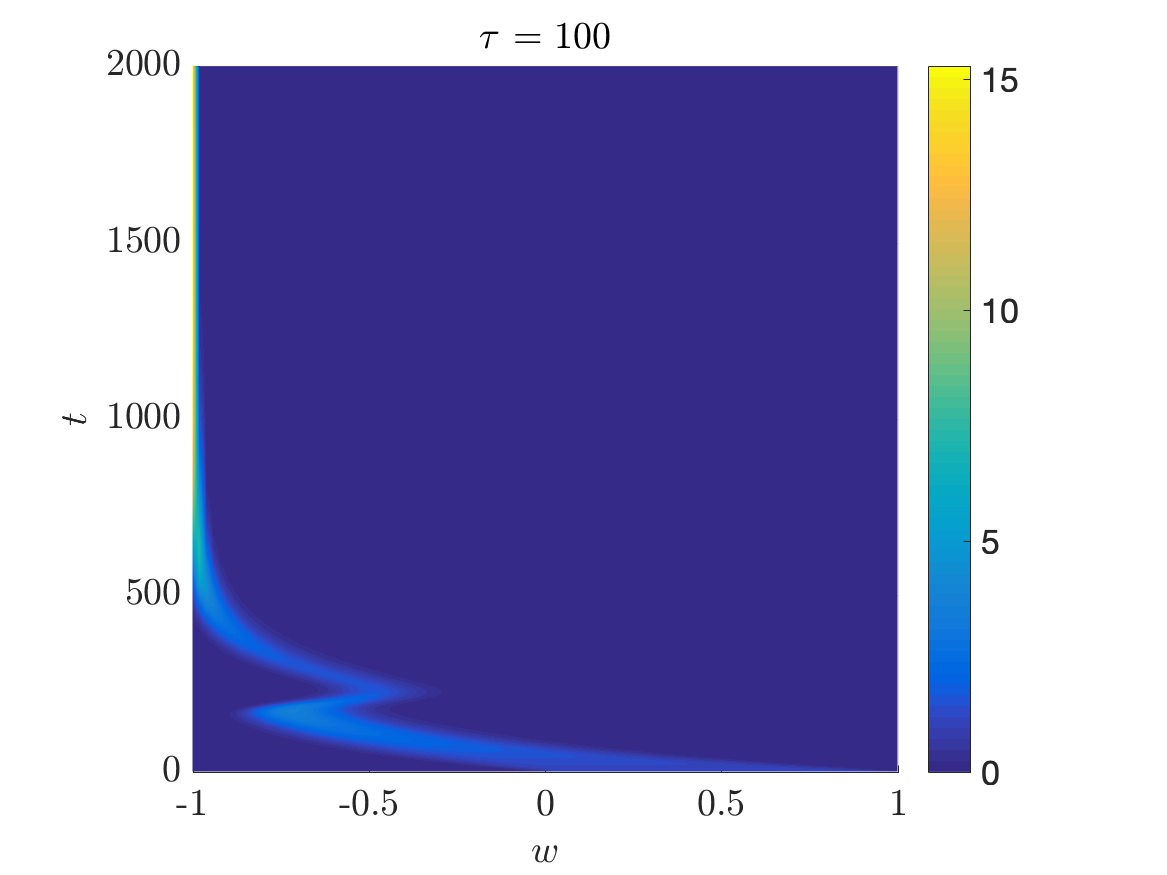}
\includegraphics[scale = 0.31]{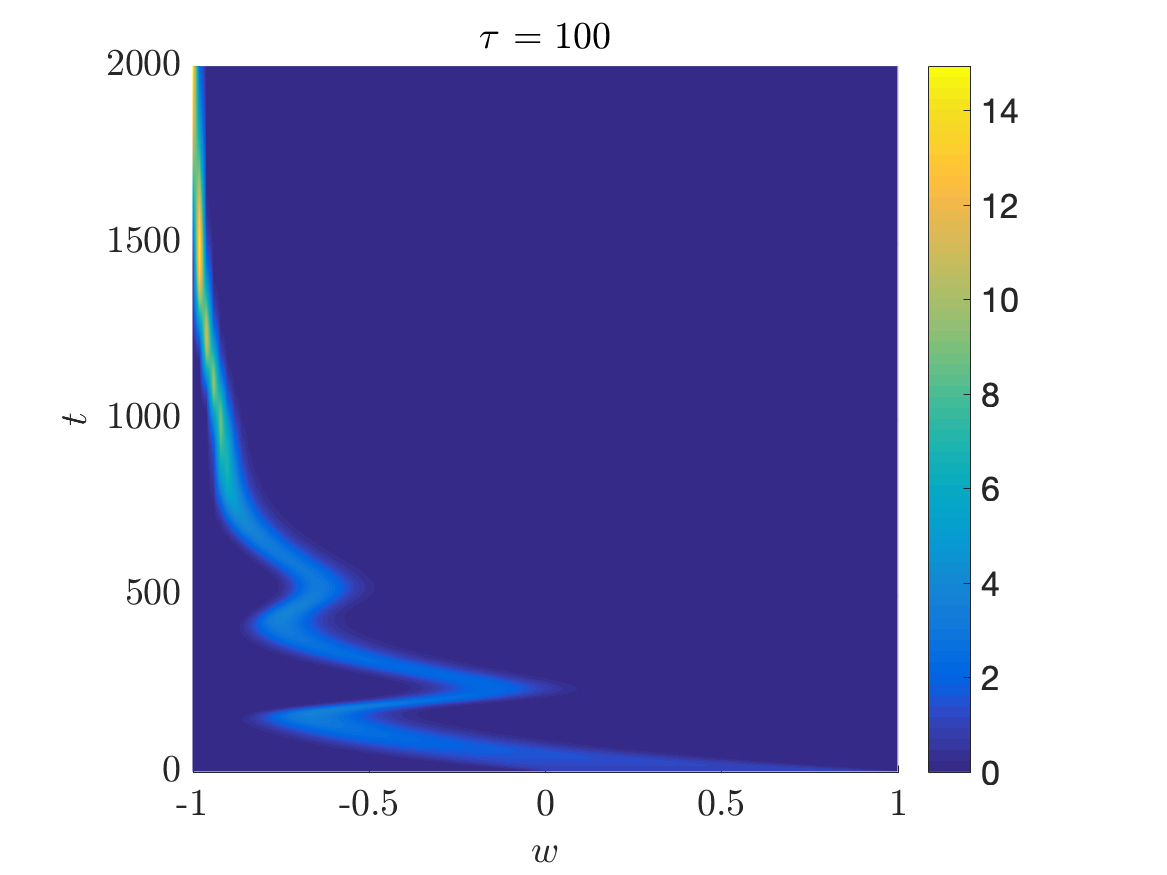}
\caption{\textbf{Test 1.} Influence of the function $\mu(\rho_I)$ on the evolution of the epidemic for different values of the exponent $n  = 4$ (left column) or $n = 20$ (right column). In both cases the numerical solution of the Fokker--Planck step has been performed through a semi-implicit SP scheme over a grid of $N = 101$ nodes and $\Delta t = 1$. The initial distribution has been defined in \fer{eq:fJ0_t1}. We represent also the evolution of $f_S(w,t)$ over $t \in [0,T]$ for two possible choices of the scale parameter $\tau = 1$ (second row) and $\tau = 10^2$ (third row). The epidemiological parameters are given in Table \ref{tab:1}. }
\label{fig:Test1}
\end{figure}

In Figure \ref{fig:Test1} we present the results of the test, where the dynamics has been considered over the time interval $[0,2000]$. The function $\mu(\rho_I)$ acts on each compartment by driving the agents' opinions toward the extreme negative opinion $w = -1$, as soon as the epidemic is contained. For low values of the parameter $n$, the effect of the Fokker-Planck operators on the dynamics is not able to balance out the epidemiological operators $K_T$ and $K_V$. Then, the kinetic system \eqref{eq:kinetic_hesitancy} behaves like a standard SEIRV model, where the initial exponential growth of $\rho_I(t)$ is followed by a decay to the disease-free equilibrium. Instead, when $n$ is chosen big enough, we observe the creation of a fluctuating behavior in the opinion of individuals, whose decision-making with respect to the containment measures changes depending on the current state of the epidemic itself. The increase in infected individuals leads to the adoption of more protective behaviors, in order to counter the spreading of the infection. On the contrary, when the infection is reduced, the population tends to develop a more negative sentiment toward the containment measures and we observe a resurgence in the number of infected, before the epidemic eventually dies out as $\rho_I(t)$ asymptotically reaches zero. This phenomenon is a typical example of epidemic waves, that our model is capable to reproduce.

%%%%%%%%%%%%%%%%%%%%%%    SUBSECTION: TEST 2    %%%%%%%%%%%%%%%%%%%%%%%%
\subsection{Test 2: Impact of leaders and vaccine hesitancy}

As second case, we consider a situation where the presence of a strong enough leaders' campaign can influence the behavior of a population and change the outcome of the epidemic in a positive manner, by delaying or reducing its impact. We first consider the case $G(w,w_*) \equiv 1$ for any $w, w_* \in \mathcal{I}$, so that the leaders' effect is solely encapsulated by their mean opinion $m_L$ and by the parameter $\lambda_L$. In particular, we show how the balance between $\lambda_J \mu(\rho_I)$ and $\lambda_L m_L$ affects the evolution of the compartmental average opinions, for suitable increasing values of the parameter $\lambda_L$. We study the cases $\lambda_L = 0.2,0.3$. Furthermore, we consider a time-independent opinion distribution of leaders of the form 
\begin{equation}
\label{eq:fL0}
f_L(w) = c_L \cdot \chi_{[0.8,1]}(w),
\end{equation}
where $c_L>0$ is a normalization constant and  $m_L = \int_{-1}^1 w f_L(w)dw=0.9$. The kinetic model defined in  \eqref{eq:kinetic_hesitancy} is integrated on the time interval $[0,2000]$ as before. 

As initial condition, we consider in Figure \ref{fig:Test2P} the case  
\begin{equation}
\label{eq:t0_test2}
f_J(w,0) = \rho_J(0) \chi_{[0,1]}(w), \qquad J \in \mathcal{C},
\end{equation}
in which the population is biased toward a positive initial opinion distribution and therefore is likely to adopt protective measures and to uptake the vaccine. Instead, in Figure \ref{fig:Test2N} we consider the opposite case
\begin{equation}
\label{eq:t0_test2b}
f_J(w,0) = \rho_J(0) \chi_{[-1,0]}(w),\qquad J \in \mathcal{C},
\end{equation}
in which the population has a negative initial opinion distribution and therefore is less likely to adopt the protective measures. As before, we also fix $\rho_J(0) = 10^{-3}$, $J \in \{ E, I, R \}$, $\rho_V(0) = 0$ and $\rho_S(0) = 1 - \rho_E(0) - \rho_I(0) - \rho_R(0)$.

\begin{figure}
\centering
\centering
\includegraphics[scale = 0.31]{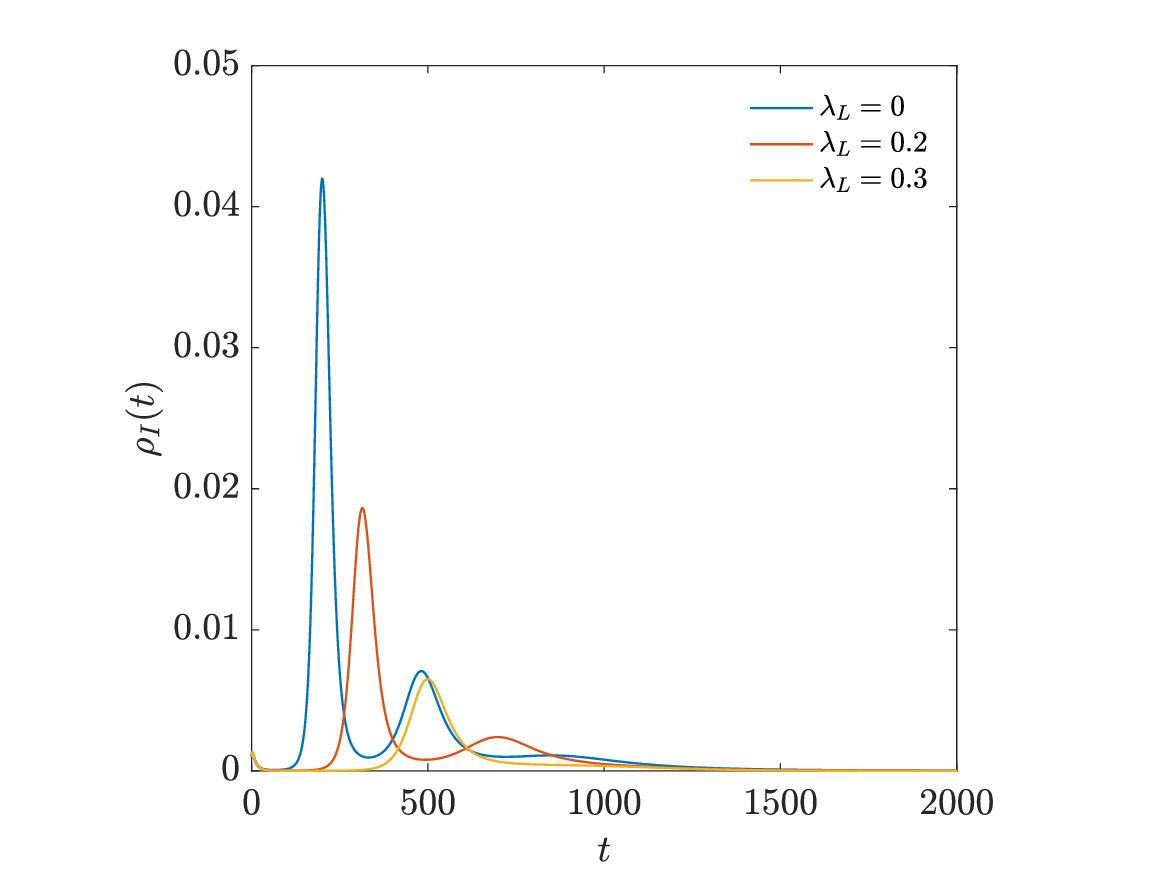} \includegraphics[scale = 0.31]{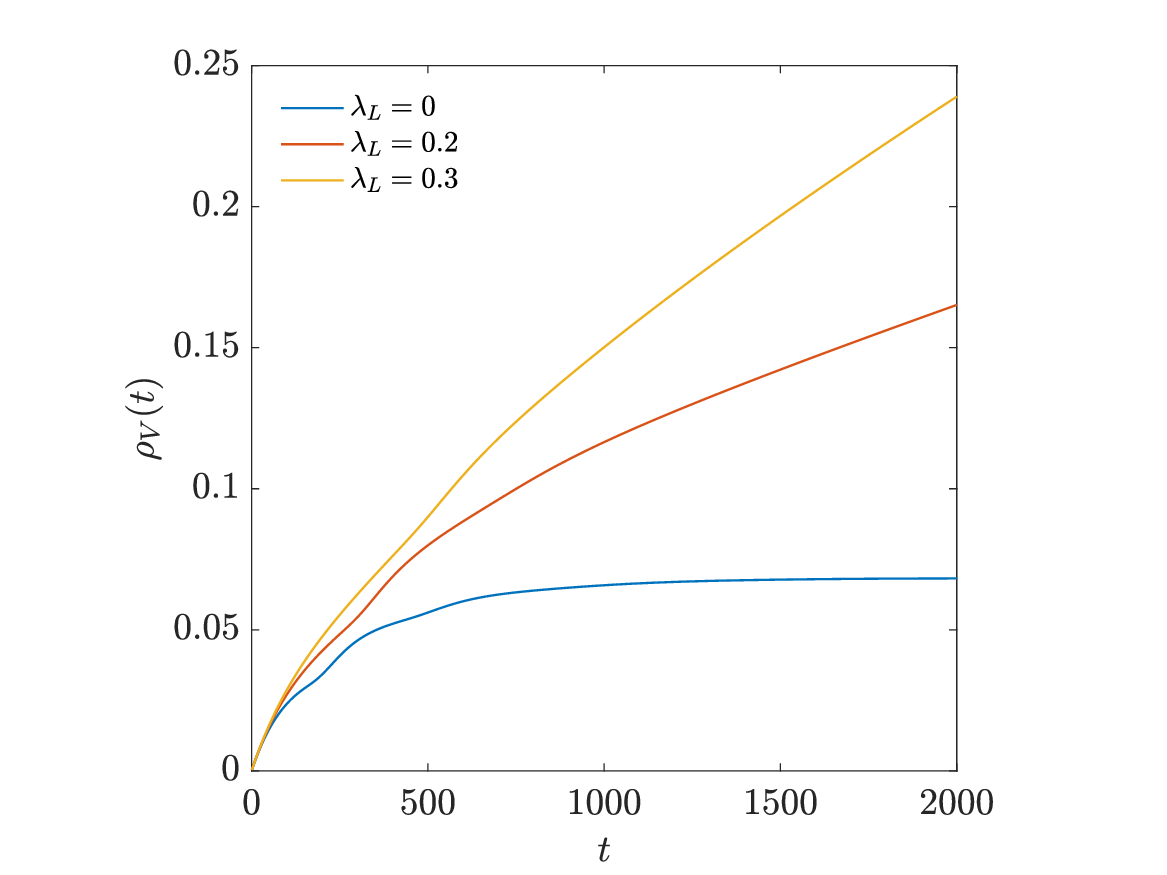}
\includegraphics[scale = 0.31]{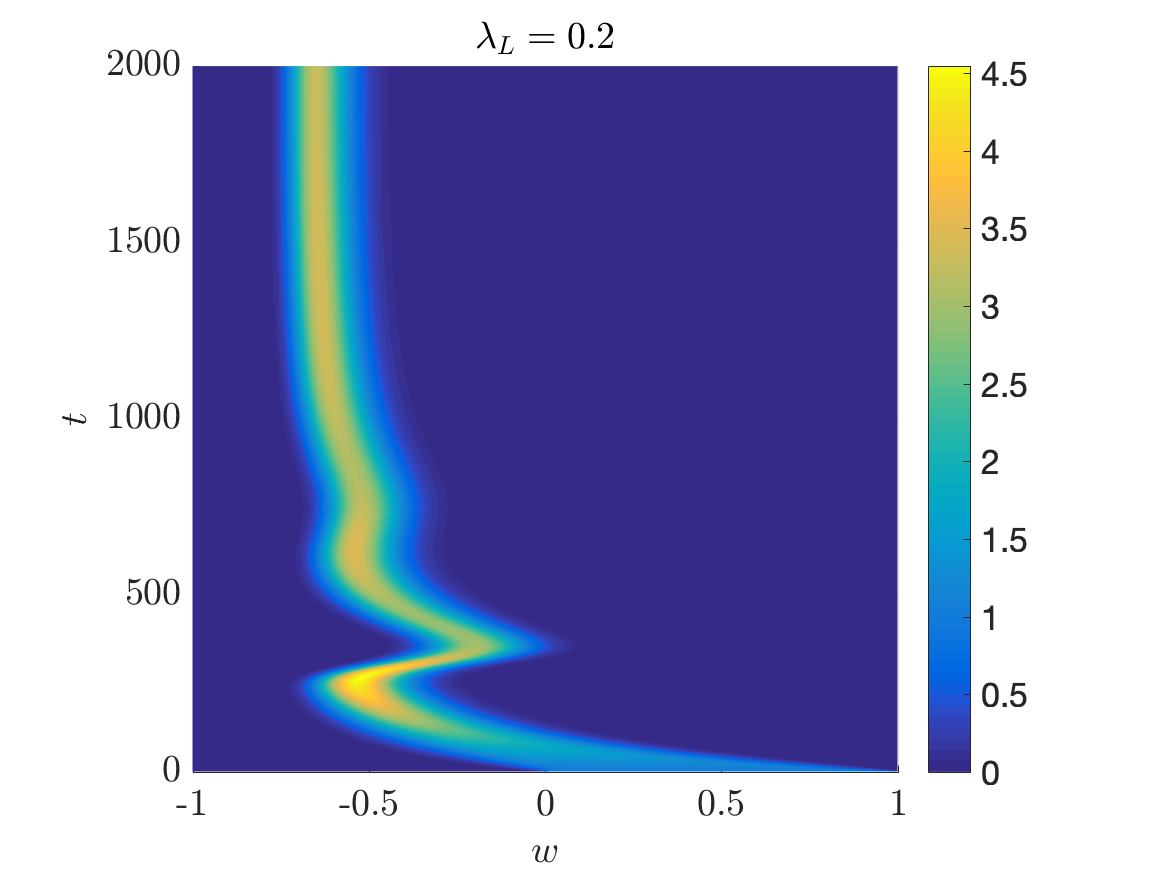} \includegraphics[scale = 0.31]{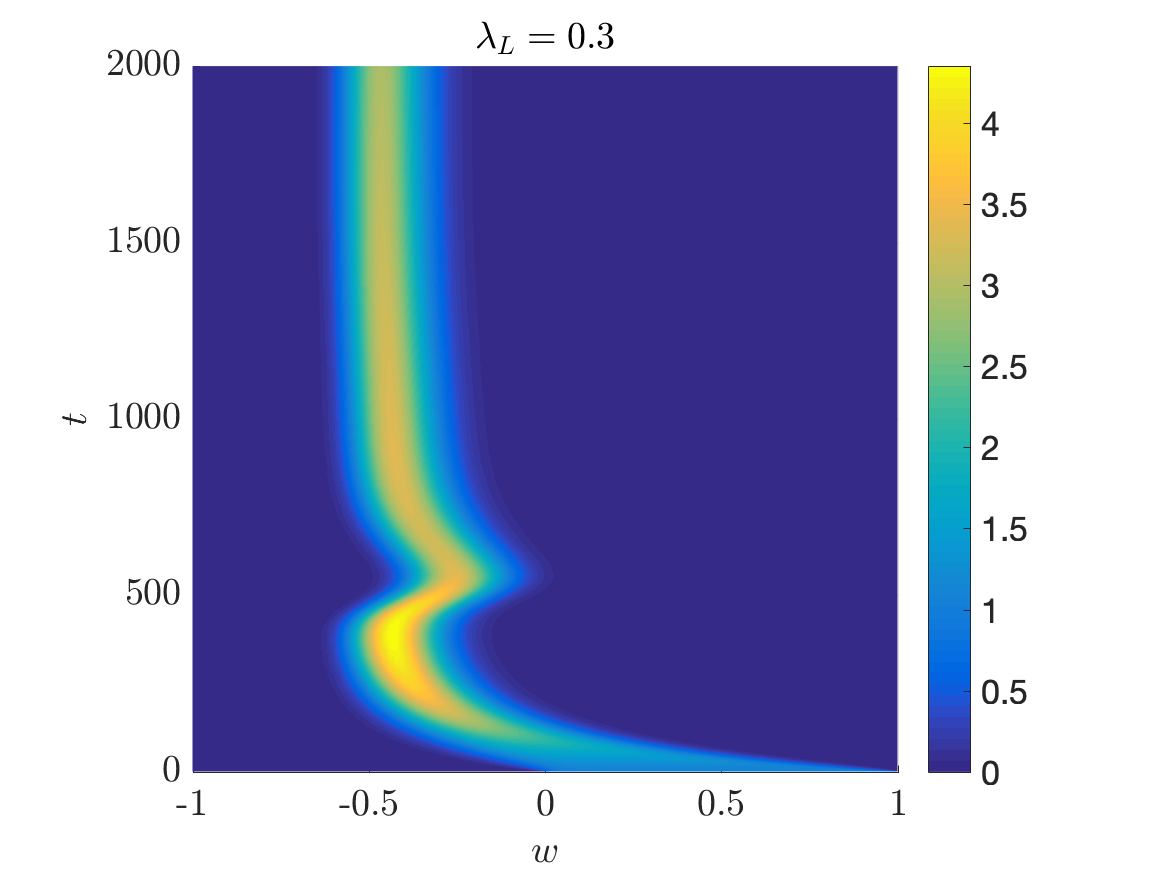}
\caption{\textbf{Test 2a}: Effects of the  interactions with the leaders on the evolution of the epidemic and on the vaccination campaign. Top row: evolution of $\rho_I(t)$ (left) and of $\rho_V(t)$ (right) over the considered time window. Bottom row: evolution of the kinetic density $f_S(w,t)$ in the cases $\lambda_L=0.2$ and $\lambda_L = 0.3$. The initial distribution has been defined in \fer{eq:t0_test2} and the solution of the Fokker--Planck step has been performed through a semi-implicit SP scheme over a grid of $N = 101$ nodes and $\Delta t = 1$. The epidemiological parameters have been defined in Table \ref{tab:1}. }
\label{fig:Test2P}
\end{figure}

\noindent The evolution of the kinetic densities $f_S$ and of the masses $\rho_I$ and $\rho_V$  are reported in Figure \ref{fig:Test2P} and in Figure \ref{fig:Test2N}. In particular, in Figure \ref{fig:Test2P} we have started from the initial condition \fer{eq:t0_test2}  and we observe, thanks to the follower-leader interactions, a progressive delay in the initial outbreak of the epidemic, together with  a  reduction in the magnitude of the infection's peaks. This effect is a natural consequence of the interaction with the leaders. In Figure \ref{fig:Test2N} we consider the initial condition \fer{eq:t0_test2b},  for which the initial opinions are negatively biased. The leaders' campaign is not effective  to delay the first epidemic's outburst, previously observed in the dynamics of $\rho_I(t)$. Anyway, it is still able to substantially reduce the peaks of the infection. In both scenarios, we also notice how the addition of a leader increases the number of vaccinated individuals $\rho_V(t)$ and drives the distribution of susceptible individuals toward a progressively more positive opinion as $\lambda_L$ increases, moving away from the extreme negative sentiment that is reached by $f_S(w,t)$ asymptotically in time when only $\mu(\rho_I)$ is present, for $\lambda_L = 0$.

\begin{figure}
\centering
\includegraphics[scale = 0.3]{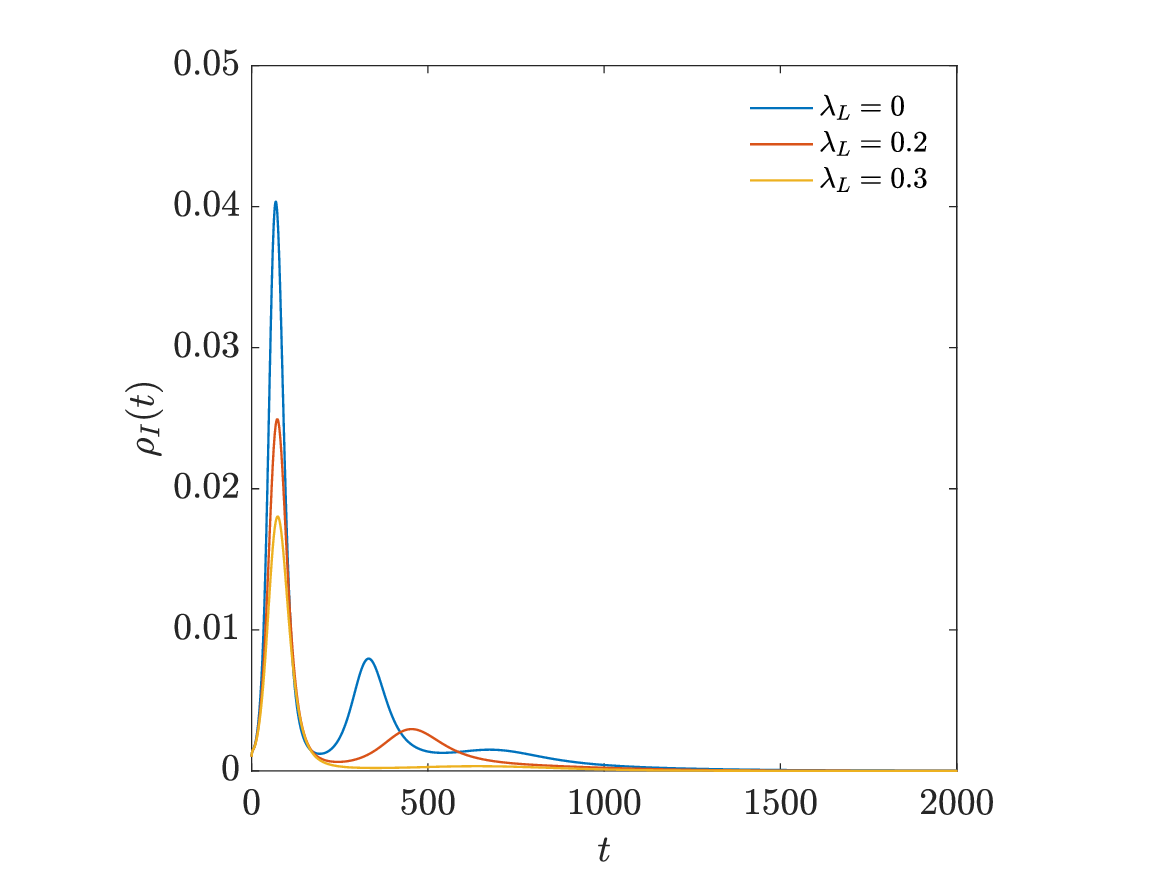} \includegraphics[scale = 0.3]{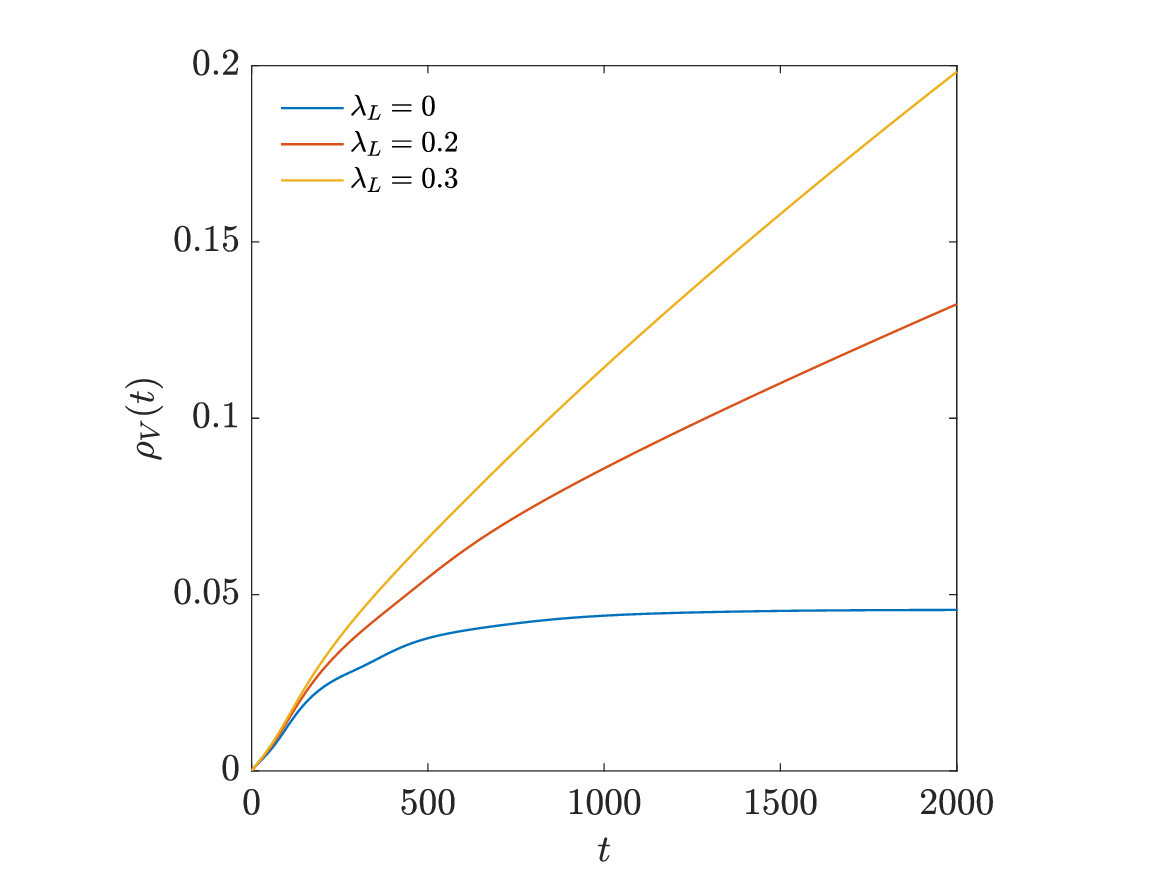}\\
\includegraphics[scale = 0.3]{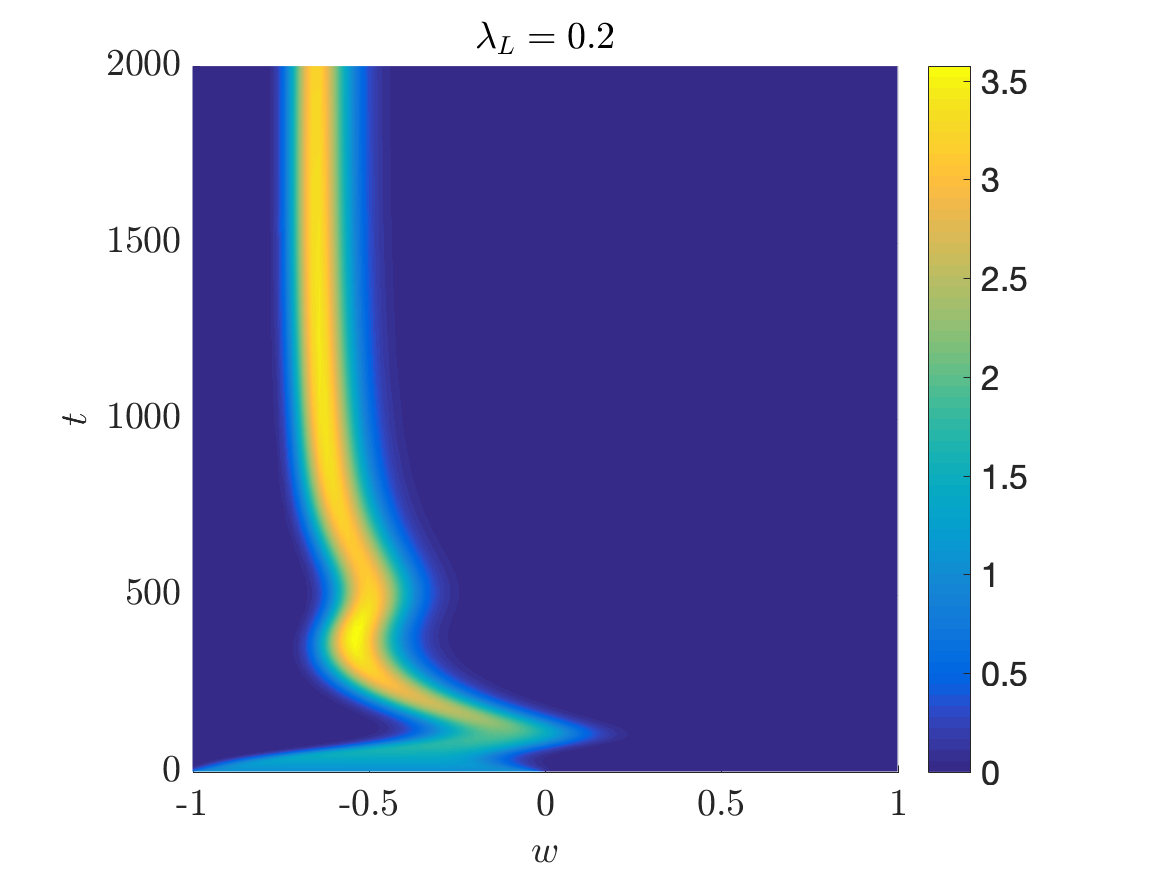} \includegraphics[scale = 0.3]{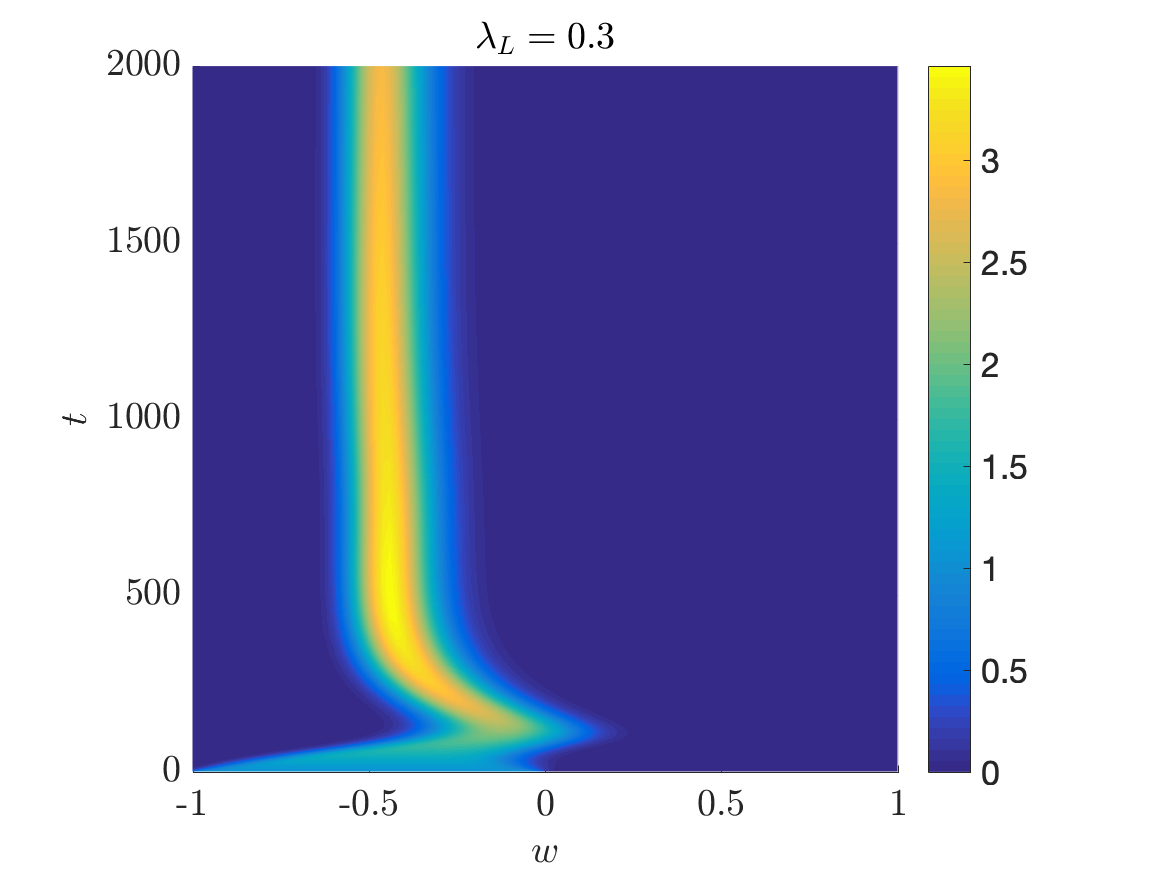}
\caption{\textbf{Test 2b}:  Effects of the  interactions with leaders on the evolution of the epidemic and on the vaccination campaign. Top row: evolution of $\rho_I(t)$ (left) and of $\rho_V(t)$ (right) over the considered time window. Bottom row: evolution of the kinetic density $f_S(w,t)$ in the cases $\lambda_L=0.2$ and $\lambda_L = 0.3$. The numerical solution of the Fokker--Planck step has been performed through a semi-implicit SP scheme over a grid of $N = 101$ nodes and $\Delta t = 1$. The initial distribution has been defined in \fer{eq:t0_test2} and is considered as positively biased. The epidemiological parameters have been defined in Table \ref{tab:1}.}
\label{fig:Test2N}
\end{figure}

%%%%%%%%%%%%%%%%%%%%%%    SUBSECTION: TEST 3    %%%%%%%%%%%%%%%%%%%%%%%%
\subsection{Test 3: Bounded confidence interactions}

We conclude with an example of bounded confidence-type dynamics, modeling the loss of global consensus inside a population via the formation of clusters in the opinion distributions of individuals \cite{HK}. Clustering effects are observed when $G$ takes the  form \eqref{BC}, so that the leaders act through the full operator $\mathcal{L}_L$ defined by \eqref{op-leader}. 

In the following, we consider the cases $\Delta = \frac{3}{2}$ and $\Delta = \frac{1}{2}$. We consider the leaders' opinion distribution $f_L(w)$ defined in \fer{eq:fL0} and we consider the case 
\begin{equation}
\begin{split}
f_J(w,0) = \rho_J(0)\dfrac{1}{2} \chi_{[-1,1]}(w), \qquad J \in \{S, E ,I, R,V\}, \\
\end{split}
\end{equation} 
with $\rho_J(0) = 10^{-3}$, $J \in \{ E, I, R \}$, $\rho_V(0) = 0$ and $\rho_S(0) = 1 - \rho_E(0) - \rho_I(0) - \rho_R(0)$. Furthermore, we focus on two values of the constant $\lambda_L = 0.2,0.3$. 

In Figure  \ref{fig:Test3a} we depict the evolution of the kinetic distribution with leader--follower interactions, weighted by a bounded confidence function with $\Delta = \frac{3}{2}$. We may observe that the presence of the population of leaders is capable to dampen the epidemic waves that appear in $\rho_I(t)$ for a sufficiently large value of $\lambda_L>0$. As a consequence, the evolution of $\rho_V(t)$ benefits from the leader--follower interactions, see the top row of  Figure \ref{fig:Test3a}. However, since the leaders cannot interact with the entire population as in the previous example, their overall containment effect is milder and gives rise to new interesting dynamics in the profile of $\rho_I(t)$, like the generation of additional waves in the case $\lambda_L = 0.2$, or the longer persistence of the epidemic over time in the case $\lambda_L = 0.3$. Similarly, we observe less vaccinated than in the previous tests and the profile of $\rho_V(t)$ exhibits a slower growth.

These phenomena may be explained by the evolution of the distribution function $f_S(w,t)$ for the considered choices of $\lambda_L = 0.2,0.3$. For $\lambda_L = 0.2$ we may notice that in the transient regime the opinions of the susceptible population tends to form two clusters due to the interactions with leaders. This behaviour is further highlighted by the case $\lambda_L = 0.3$, where we observe the formation of separate clusters. 

\begin{figure}
\centering
\includegraphics[scale = 0.3]{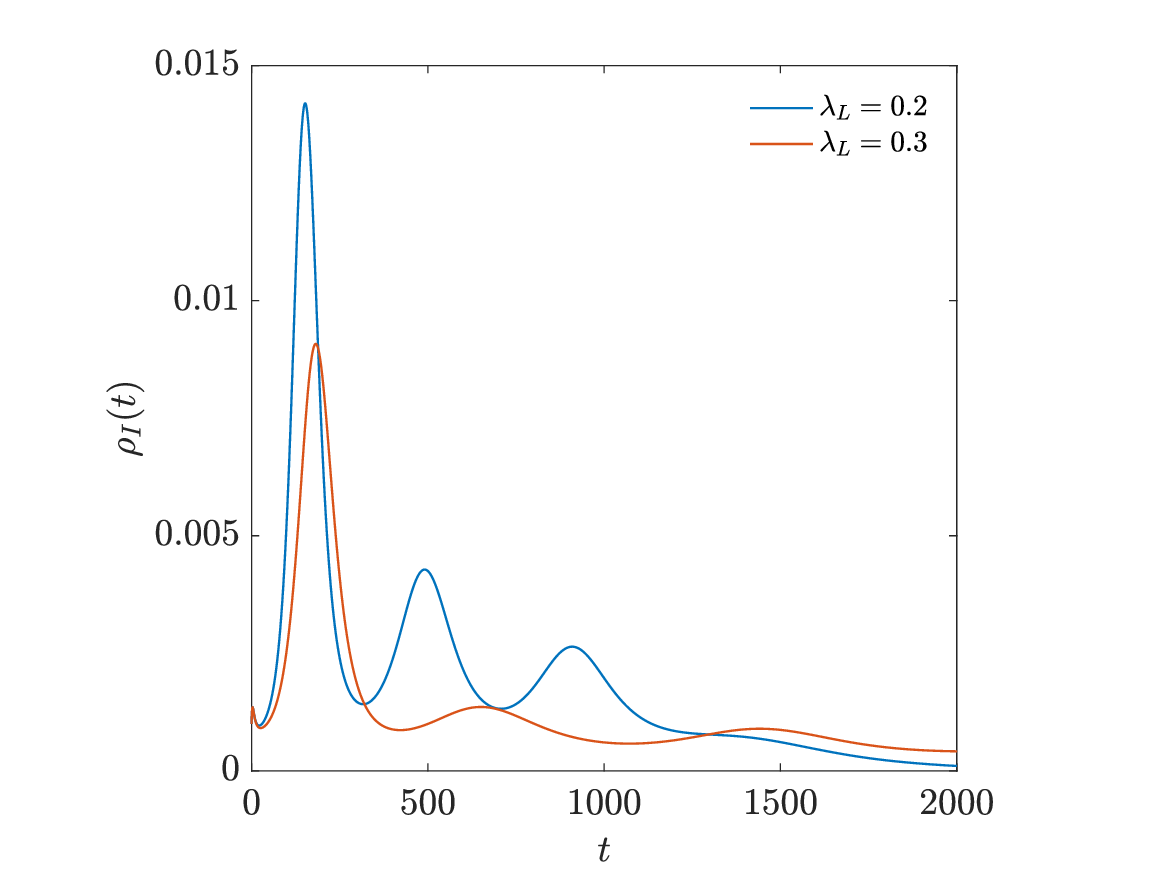} \includegraphics[scale = 0.3]{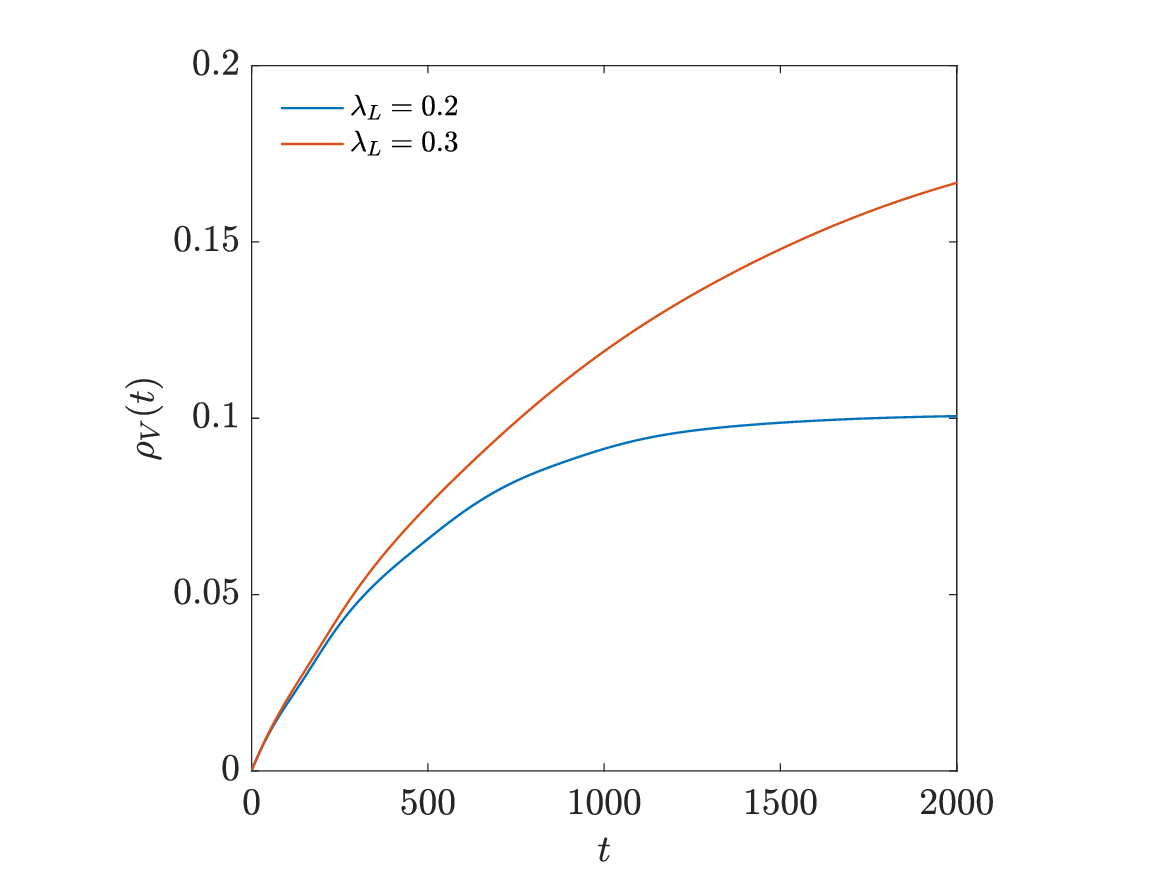}
\includegraphics[scale = 0.3]{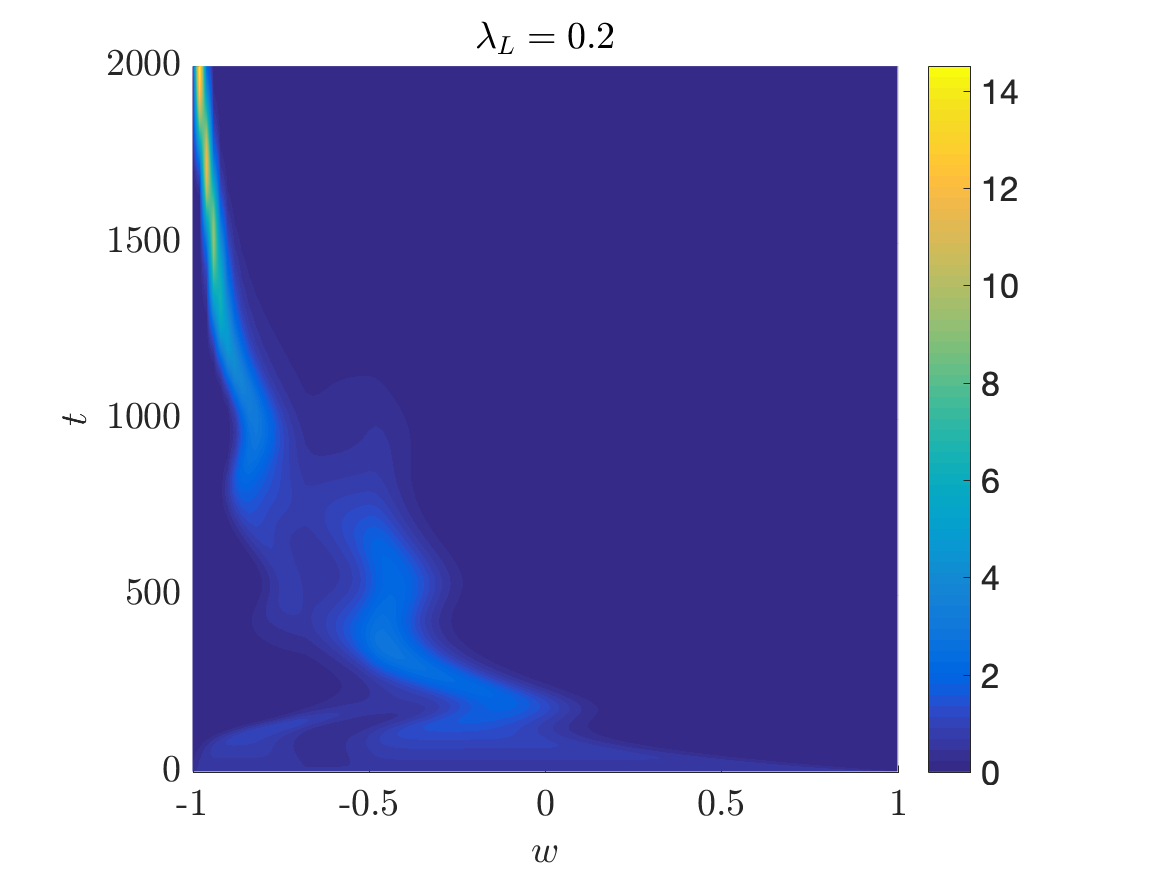} \includegraphics[scale = 0.3]{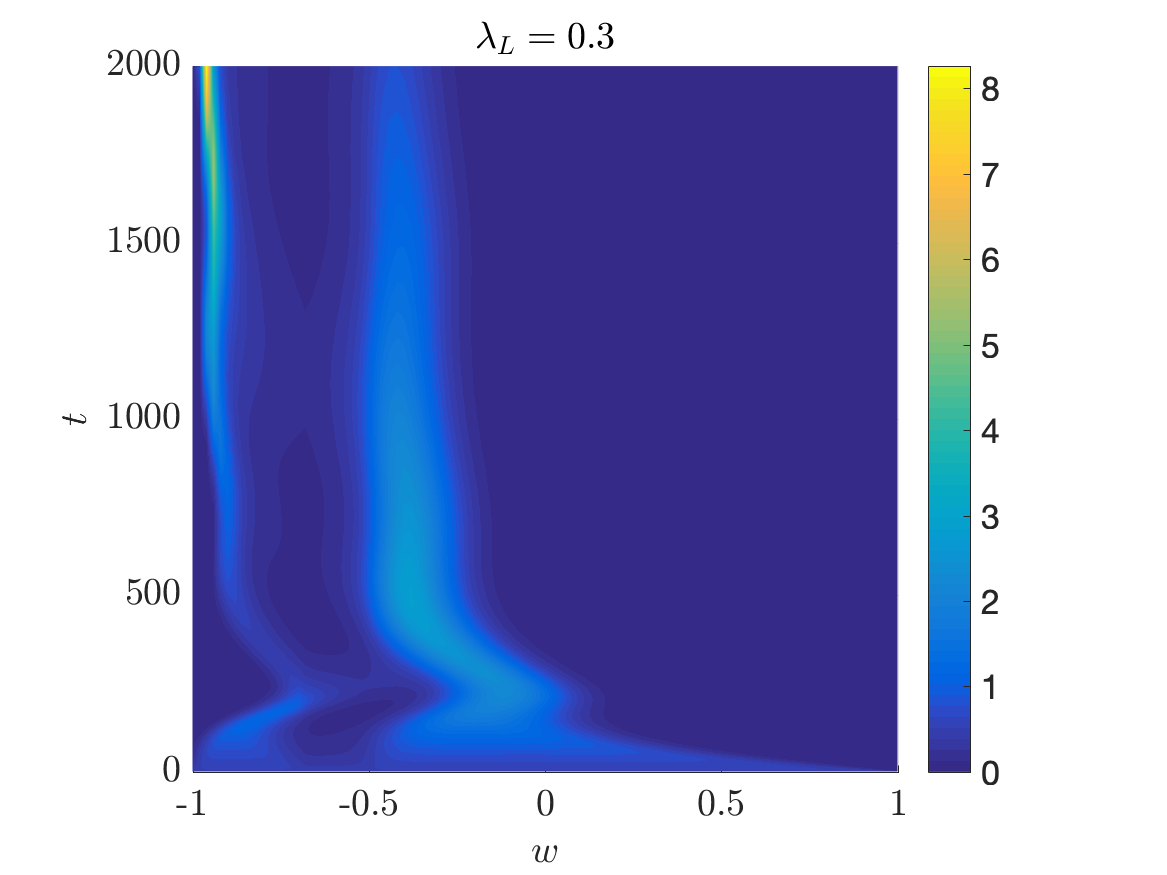}
\caption{\textbf{Test 3a}:  Effects of the interactions with leaders on the evolution of the epidemic and on the vaccination campaign. The leader--follower dynamics is described by \fer{op-FPH}, with bounded confidence interaction function  $G(\cdot,\cdot)$ as in \fer{BC} and $\Delta = \frac{3}{2}$. Top row: evolution of $\rho_I(t)$ (left) and of $\rho_V(t)$ (right) over the considered time window. Bottom row: evolution of the kinetic density $f_S(w,t)$ in the cases $\lambda_L=0.2$ and $\lambda_L = 0.3$. The numerical solution of the Fokker--Planck step has been performed through a semi-implicit SP scheme over a grid of $N = 101$ nodes and $\Delta t = 1$. The initial distribution has been defined in \fer{eq:t0_test2} and is considered as positively biased. The epidemiological parameters have been defined in Table \ref{tab:1}.} 
\label{fig:Test3a}
\end{figure}

In Figure \ref{fig:Test3b} we depict the evolution of the kinetic distribution with leader--follower interactions, weighted by a bounded confidence function with $\Delta = \frac{1}{2}$. In this case, the leaders can interact with a smaller fraction of the susceptible population, so that their influence is not enough to prevent the formation of epidemic waves. In particular, the evolution of the total number of vaccinated individuals displays small benefits from the presence of opinion leaders.

\begin{figure}
\centering
\includegraphics[scale = 0.3]{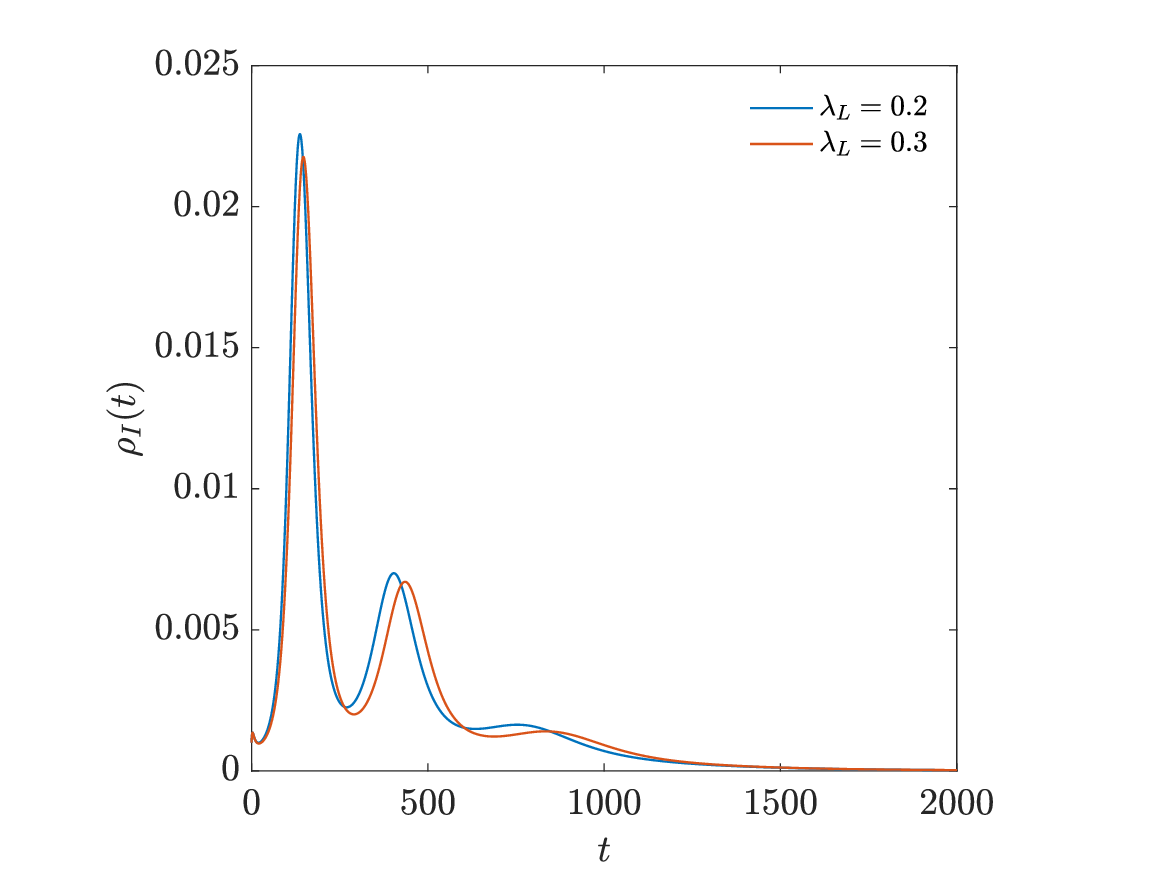} \includegraphics[scale = 0.3]{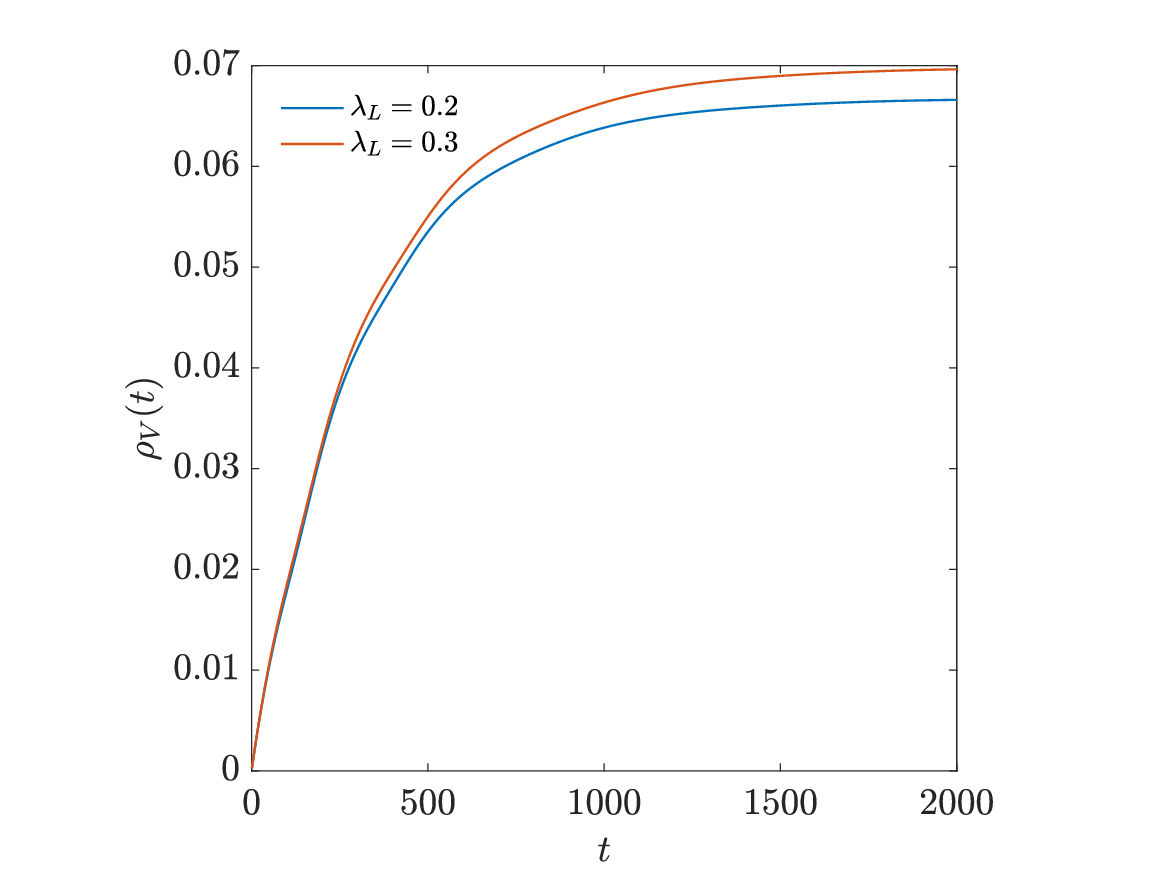}\\
\includegraphics[scale = 0.3]{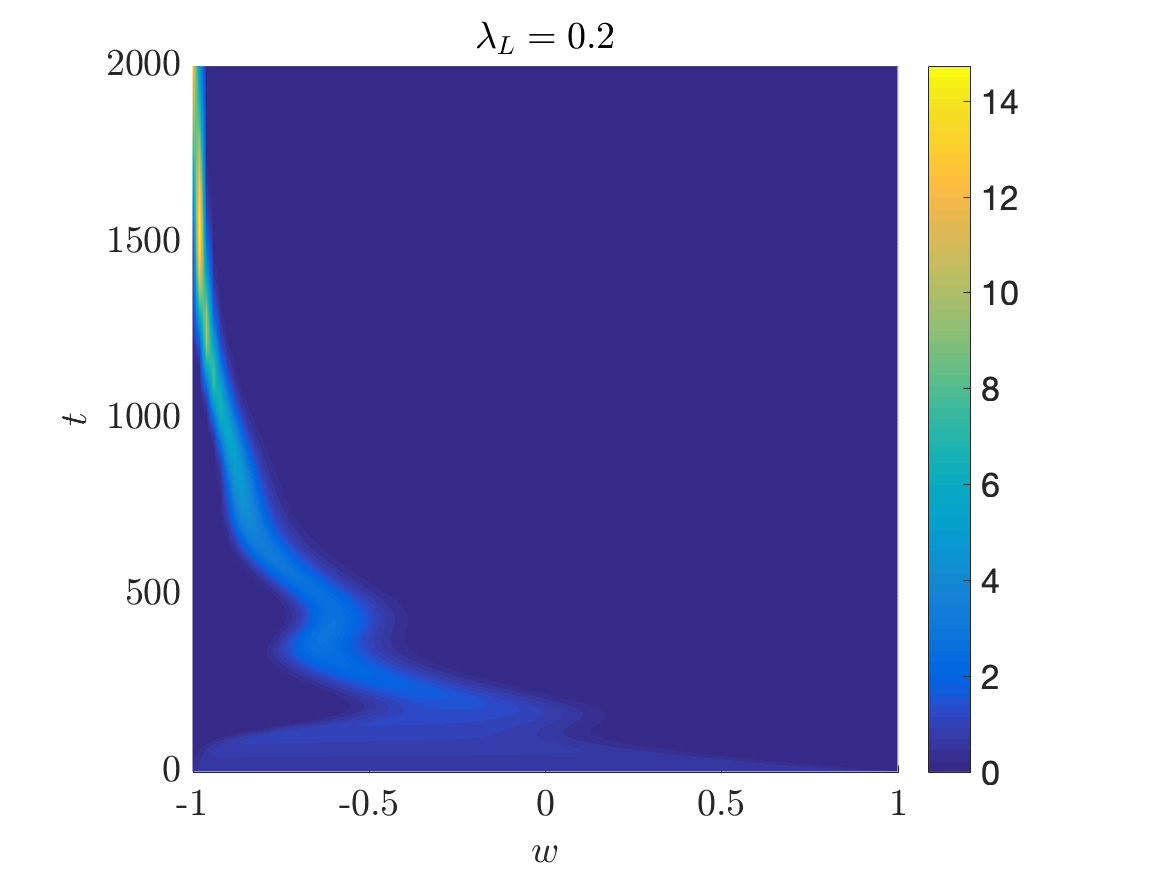} \includegraphics[scale = 0.3]{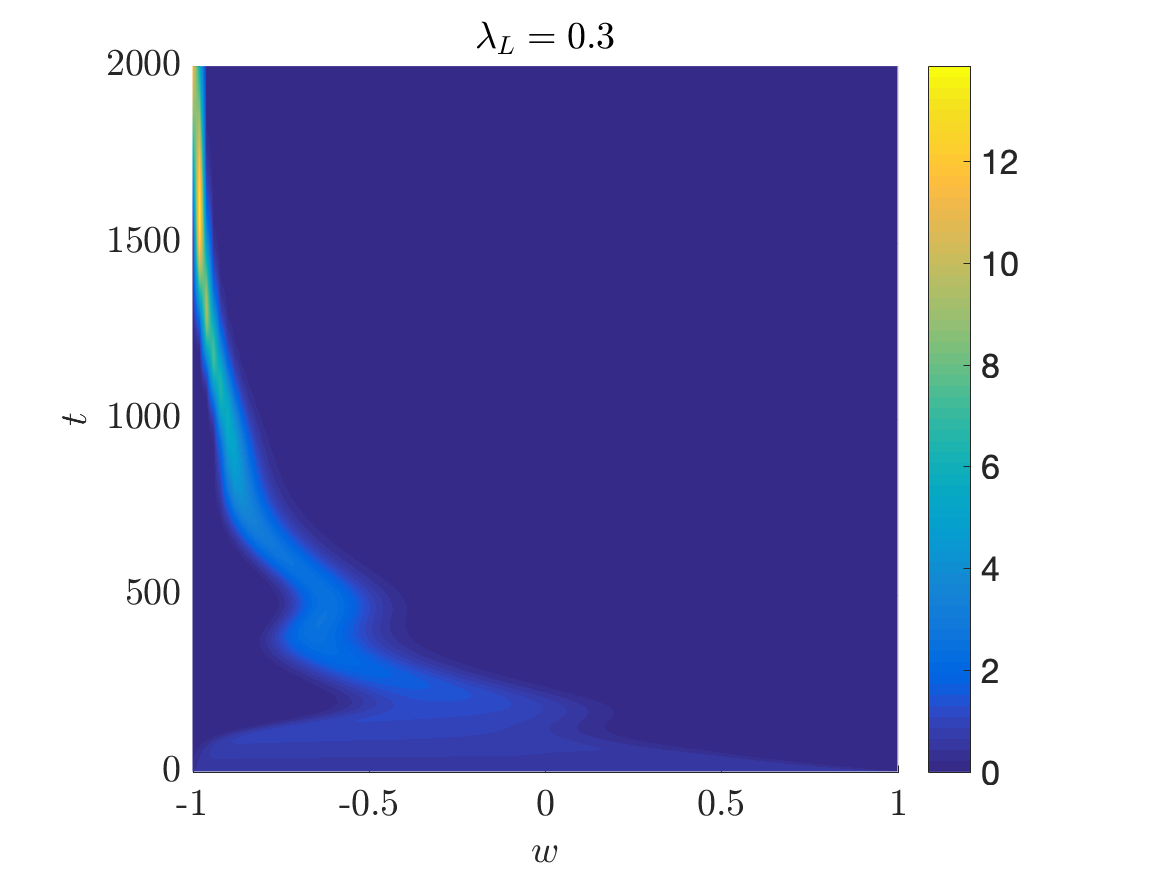}
\caption{\textbf{Test 3b}:  Effects of the  interactions with leaders on the evolution of the epidemic and on the vaccination campaign. The leader--follower dynamics is described by \fer{op-FPH}, with bounded confidence interaction function  $G(\cdot,\cdot)$ as in \fer{BC} and $\Delta = \frac{1}{2}$. Top row: evolution of $\rho_I(t)$ (left) and of $\rho_V(t)$ (right) over the considered time window. Bottom row: evolution of the kinetic density $f_S(w,t)$ in the cases $\lambda_L=0.2$ and $\lambda_L = 0.3$. The numerical solution of the Fokker--Planck step has been performed through a semi-implicit SP scheme over a grid of $N = 101$ nodes and $\Delta t = 1$. The initial distribution has been defined in \fer{eq:t0_test2} and is considered as positively biased. The epidemiological parameters have been defined in Table \ref{tab:1}.}
\label{fig:Test3b}
\end{figure}

%%%%%%%%%%%%%%%%%%%%%%%%%%%%%%%%%%%%%%%%%%%%%%%%%%%%%%%%%%%%
%%%%%%%%%%%%%%%%%%%%%%%%%%%%%%%%%%%%%%%%%%%%%%%%%%%%%%%%%%%%
%%%%%%%%%%%%%%%%%%%%    SECTION 5: CONCLUSION    %%%%%%%%%%%%%%%%%%%%%%%
%%%%%%%%%%%%%%%%%%%%%%%%%%%%%%%%%%%%%%%%%%%%%%%%%%%%%%%%%%%%
%%%%%%%%%%%%%%%%%%%%%%%%%%%%%%%%%%%%%%%%%%%%%%%%%%%%%%%%%%%%
\section*{Conclusion}
Among other social aspects of opinion formation which deserve to be investigated, vaccine hesitancy is one of the most important, in reason  of its implications in improving social health \cite{BDMOG,BRBU,DMLM,Tsao,Trentini}. 
Recently, this phenomenon has been shown to be in close relationship with a positive evolution of the COVID-19 pandemic (cf.  Ref. \cite{Kum} and the references therein).

Following the recent contributions of Refs. \cite{DPaTZ,DPTZ,DTZ,Z}, in this paper we have studied the evolution of the distribution function  of the vaccine hesitancy of a population in presence of the spreading of a disease, described through  a classical compartmental model \cite{DH}. The coupling disease--hesitancy has been considered to act in both directions. From one side, it has been assumed that the presence of the disease leads in general to a decrease of the vaccine hesitancy  of individuals, both because of the perceived increase in risks associated with exposure to the infection \cite{DC,VBA}, and because of the presence of a possible pro-vaccination campaign, characterized in this work by so-called opinion leaders \cite{DMPW,DW}. On the other side,  phases of stagnation or decline in the epidemic spreading naturally lead to a growth in the vaccine hesitancy of the population, modeled here by acting on the various compartmental average vaccination rates, through the assumption of a dependency on the number of infected individuals.

Numerical experiments have confirmed the ability of the model to reproduce dampened epidemiological waves due to differing risk perceptions, which follow from a varying intensity of the disease and result in a variability of the vaccine hesitancy. Moreover, we have tested and quantified in various situations the presence of leaders, which leads to an improvement in the vaccine hesitancy of the population, opening the way to make use of the  model for possible forecasts. Future studies will aim to balance the parameters of the model, by resorting to existing experimental data.

\vskip 3cm
\noindent \textbf{Acknowledgements.} This work has been written within the activities of GNFM group of INdAM (National Institute of High Mathematics). AB and MZ acknowledge the support of MUR-PRIN2020 Project No.2020JLWP23 (Integrated Mathematical Approaches to Socio-Epidemiological Dynamics). AB also acknowledges partial support from the Austrian Science Fund (FWF), within the Lise Meitner project number M-3007 (Asymptotic Derivation of Diffusion Models for Mixtures). GT wishes to acknowledge partial support by IMATI, Institute for Applied Mathematics and Information Technologies “Enrico Magenes”, Via Ferrata 5 Pavia, Italy.

\newpage
%%%%%%%%%%%%%%%%%%%%%%%%%%%%%%%%%%%%%%%%%%%%%%%%%%%%%%%%%%%%
%%%%%%%%%%%%%%%%%%%%%%%%%%%%%%%%%%%%%%%%%%%%%%%%%%%%%%%%%%%%
%%%%%%%%%%%%%%%%%%%%%%%%    BIBLIOGRAPHY    %%%%%%%%%%%%%%%%%%%%%%%%%
%%%%%%%%%%%%%%%%%%%%%%%%%%%%%%%%%%%%%%%%%%%%%%%%%%%%%%%%%%%%
%%%%%%%%%%%%%%%%%%%%%%%%%%%%%%%%%%%%%%%%%%%%%%%%%%%%%%%%%%%%

\setlength\parindent{0pt}

\end{document}